\documentclass{amsart}
\usepackage{float}
\usepackage{geometry}
\usepackage{array}
\usepackage{booktabs}
\usepackage{multirow}
\usepackage{amsmath}
\usepackage{amsthm}
\usepackage{xcolor}

\newif\ifcleanversion
\cleanversiontrue 

\ifcleanversion
  \RenewDocumentCommand{\textcolor}{o m +m}{#3}
  \NewDocumentCommand{\sout}{+m}{}
\else
  \usepackage[normalem]{ulem}
\fi

\usepackage{graphicx}
\usepackage{setspace}
\usepackage[foot]{amsaddr}
\usepackage{tikz}
\usetikzlibrary{quantikz2}
\usepackage{longtable}

\makeatletter

\usepackage{thmtools}
\usepackage{thm-restate}

\providecommand{\tabularnewline}{\\}
\floatstyle{ruled}
\newfloat{algorithm}{H}{loa}
\providecommand{\algorithmname}{Algorithm}
\floatname{algorithm}{\protect\algorithmname}

\theoremstyle{definition}
\newtheorem{defn}{\protect\definitionname}
\theoremstyle{plain}
\newtheorem{lem}{\protect\lemmaname}
\theoremstyle{plain}

\theoremstyle{definition}
 \newtheorem{example}{\protect\examplename}
\theoremstyle{plain}
\newtheorem{prop}{\protect\propositionname}
\theoremstyle{plain}

\usepackage{amssymb}
\usepackage{caption}
\usepackage{tikz-cd}

\usepackage{indentfirst}
\usepackage{ragged2e}

\usepackage{algpseudocode}

\@ifundefined{showcaptionsetup}{}{%
\PassOptionsToPackage{caption=false}{subfig}}
\usepackage{subfig}
\makeatother

\providecommand{\corollaryname}{Corollary}
\providecommand{\definitionname}{Definition}
\providecommand{\examplename}{Example}
\providecommand{\lemmaname}{Lemma}
\providecommand{\propositionname}{Proposition}
\providecommand{\theoremname}{Theorem}

\usepackage{hyperref}

\begin{document}
\title{Molecular Ground State Simulation by Subspace Restriction and Hund's
Rule}
\author{Tsung-Chi Chiang$^{1,2}$}
\email{lc1026@ncts.ntu.edu.tw}

\author{Calvin Ku$^{3,4}$}
\email{calvin.ku@foxconn.com}

\author{Jyh-Pin Chou$^{5,6}$}
\email{jpchou@ntu.edu.tw}

\author{Alice Hu$^{4,7}$}
\email{alicehu@cityu.edu.hk}

\author{Peng-Jen Chen$^{4}$}
\email{pjchen1015@gmail.com}

\author{Ching-Jui Lai$^{1,2,8}$}
\email{cjlai72@mail.ncku.edu.tw}

\address{$^1$Department of Mathematics, National Cheng Kung University, Tainan 70101, Taiwan}

\address{$^2$Math Division, National Center for Theoretical Sciences, National Taiwan University, Taipei 106319, Taiwan}

\address{$^3$Hon Hai Research Institute, Taipei, Taiwan}

\address{$^4$Department of Mechanical Engineering, City University of Hong Kong, Kowloon, Hong Kong SAR 999077, China}

\address{$^5$Graduate School of Advanced Technology, National Taiwan University, Taipei 106319, Taiwan}
\address{$^6$Physics Division, National Center for Theoretical Sciences, National Taiwan University, Taipei 106319, Taiwan}

\address{$^7$Department of Material Science and Engineering, City University of Hong Kong, Kowloon, Hong Kong SAR 999077, China}

\address{$^8$Corresponding Author}

\begin{abstract}
\textcolor{red}{The Subspace Restriction Scheme (SRS) is presented to address the qubit
and resource constraints of molecular simulation on near-term quantum
hardware. SRS restricts a molecular Hamiltonian to a selected
determinant subspace before qubit encoding, so a $D$-dimensional
subspace requires only $q=\lceil\log_2D\rceil$ qubits. We apply SRS to
the established seniority-resolved Multi-Hund Subspaces
$\mathrm{MHS}(\Omega)$ and develop compact encodings for these sectors.
At fixed seniority and filling fraction, MHS asymptotically uses about
half as many qubits as a compact particle-conserving encoding. For the
$\mathrm{H}_{22}$ chain, the register shrinks from $44$ to $20$ qubits.
Compression can nevertheless increase Pauli-operator density and does
not always reduce measurement cost. In noiseless VQE benchmarks,
shallow circuits reach within $1.6$~mHa of the exact subspace minima using only $4$--$30$ CNOT gates. However, the corresponding subspace minima for LiH, HF, OH, and CH remain $4.39$--$9.73$~mHa above FCI. Full-space UCCSD reaches chemical accuracy but requires hundreds to thousands of CNOT gates. SRS therefore provides a controlled resource--accuracy
trade-off when register width is the binding constraint and the target
state is well represented within a selected seniority sector.}
\end{abstract}
\maketitle

\section{Introduction}
The ground state energy and corresponding wavefunction of molecular systems are fundamental to numerous applications in chemistry and materials science. Specifically, ground state energy is directly related to reaction rates \cite{ReactionRates_1,ReactionRates_2}, molecular properties \cite{MolecularProperties}, and energy spectra \cite{EnergySpectra}. Additionally, the ground state wavefunction determines magnetic phenomena such as superparamagnetism \cite{superparamagnetic_behavior}. Computationally, simulating the molecular ground state involves diagonalizing the molecular Hamiltonian, where the lowest eigenvalue corresponds to the ground-state energy and the associated eigenvector represents the ground-state wavefunction.

Although ground-state simulation is classified as QMA-complete \cite{QMA_Complete,QMA_GroundState}, quantum computers offer a promising approach for simulating molecular systems that are intractable by classical diagonalization. The second-quantized formulation provides a natural many-body representation of molecular Hamiltonians, which can be mapped onto qubits using transformations such as the Jordan–Wigner (JW) and Bravyi–Kitaev \cite{BK_and_JW, Bravyi_Kitaev_Transformation}. However, on near-term quantum hardware \cite{preskill2018quantum}, the feasibility of these simulations is limited by the number of available qubits, circuit depth, and measurement overhead.

A basic physical simplification comes from particle conservation. Since the molecular Hamiltonian preserves electron number, one may restrict the full Fock space to the subspace with a fixed number of particles. For a system with $M$ spatial orbitals and $N$ electrons\textcolor{red}{\sout{ in a minimal basis set}}, this particle-conserving subspace contains $\binom{2M}{N}$ basis states; therefore, the corresponding theoretical minimum qubit requirement is $\left\lceil \log_{2}\binom{2M}{N}\right\rceil$ \cite{2ndQFO}.

Qubit Efficiency Encoding (QEE) \cite{QEE+singlet, Randomized_Linear_Encoding} offers a general method for reducing Hamiltonian size by selecting a target subspace prior to encoding. However, the restricted Hamiltonian need not preserve the ground-state energy and wavefunction of the original Hamiltonian. These facts motivate our study of the following key question: whether further qubit reduction is possible by imposing additional physically motivated constraints, and how such constraints influence the accuracy of the resulting ground-state description.

To address these questions, this work formulates a mathematical framework for Hamiltonian reduction, termed the \emph{Subspace Restriction Scheme} (SRS). The scheme consists of four steps:
\begin{enumerate}
    \item obtain the second-quantized Hamiltonian of a molecular system;
    \item restrict this Hamiltonian to a selected Fock subspace;
    \item define an encoding from the restricted Fock subspace into qubit space and transform the restricted Hamiltonian accordingly;
    \item decompose the encoded Hamiltonian into generalized Pauli operators, or more generally into a chosen set of unitary operators.
\end{enumerate}
Within this framework, the second step primarily drives both qubit reduction and spectral deviation. While it substantially reduces the effective Fock-space dimension, it may also exclude configurations relevant to the exact ground state.

\textcolor{red}{Meanwhile, two additional physical constraints are imposed beyond particle conservation. The first constraint is molecular multiplicity: the spin projection $M_s=(N_\alpha-N_\beta)/2$, where $N_\alpha$ and $N_\beta$ denote the total electron numbers for $\alpha$ and $\beta$ spin electrons, is fixed to the value $S$ determined by the target multiplicity $2S+1$. The second constraint is a generalized Hund's rule, which requires that each spatial orbital is occupied by a spin-up electron before a spin-down electron, ensuring that no orbital contains only a spin-down electron. Defining the seniority $\Omega$ as the number of singly-occupied spatial orbitals, the Hund rule mandates that all $\Omega$ unpaired electrons are spin-parallel. Consequently, a Hund configuration with seniority $\Omega$ achieves the maximal spin projection available at that seniority, $M_s=\Omega/2$.}

\textcolor{red}{These constraints define four related subspaces within the particle-conserving space: the Particle Conservation Subspace (PCS); the Multiplicity Subspace (MS), defined as the fixed-$M_s$ subspace; the Hund Subspace (HS), which enforces the generalized Hund occupation rule; and the seniority-resolved Multiplicity Hund Subspace $\mathrm{MHS}(\Omega)$, which applies both restrictions. The following relationships hold: $\mathrm{MHS}(\Omega)\subseteq\mathrm{HS}\subseteq\mathrm{PCS}$ and $\mathrm{MHS}(\Omega)\subseteq\mathrm{MS}\subseteq\mathrm{PCS}$. Figure~\ref{fig:Subspace-Relation} depicts these relationships, and Table~\ref{table: example of qubit usage} provides representative qubit counts.}


\textcolor{red}{For each admissible value of $\Omega$,
$\mathrm{MHS}(\Omega)$ is defined by the span of all determinants in HS having seniority
$\Omega$. Thus,
\[
    \mathrm{HS}=\bigoplus_{\Omega}\mathrm{MHS}(\Omega).
\]
When $\Omega=0$, every spatial orbital is either empty or doubly occupied,
and $\mathrm{MHS}(0)$ coincides exactly with the seniority-zero doubly
occupied configuration interaction (DOCI) space. For $\Omega>0$,
$\mathrm{MHS}(\Omega)$ is the corresponding open-shell seniority sector
subject to the generalized Hund occupation rule. Under a common spatial-
orbital basis for the $\alpha$ and $\beta$ orbitals, this rule selects the
highest-weight spin sector with $S=M_s=\Omega/2$. These seniority-restricted spaces and related paired-wavefunction parametrizations are established in the DOCI, pCCD/AP1roG, AGP, and
seniority-based coupled-cluster research
\cite{Ayers2025DOCI,Henderson2014SeniorityCC,
HendersonChenScuseria2026,Bytautas2011,Stein2014,Coleman1965AGP}.}




\begin{figure}[h]
\centering
\begin{tikzpicture}[scale=1.5]
\draw (0,0) circle (2cm);
\node at(180:1.2cm) {(4)=PCS};
\draw (340:0.6cm) circle (1.1cm);
\node at(320:1cm) {(3)=MS};
\draw (70:0.9cm) circle (1cm);
\node at(80:1.2cm) {(2)=HS};
\node at(45:0.5cm) {(1)=MHS};
\end{tikzpicture} \caption{\emph{Subspace Relation}.
\label{fig:Subspace-Relation}}
\end{figure}

\begin{table}[H]
\begin{centering}
\begin{tabular*}{10cm}{@{\extracolsep{\fill}}cccccc}
\toprule
\multirow{1}{*}{{\footnotesize{}Molecule}} & {\footnotesize{}MHS} & {\footnotesize{}HS} & {\footnotesize{}MS} & {\footnotesize{}PCS} & {\footnotesize{}JW}\tabularnewline
\midrule
{\footnotesize{}$\mathrm{HeH^{+}}$} & {\footnotesize{}1} & {\footnotesize{}1} & {\footnotesize{}2} & {\footnotesize{}2} & {\footnotesize{}4}\tabularnewline
{\footnotesize{}$\mathrm{HF}$} & {\footnotesize{}3} & {\footnotesize{}5} & {\footnotesize{}6} & {\footnotesize{}7} & {\footnotesize{}12}\tabularnewline
{\footnotesize{}$\mathrm{H_{2}O}$} & {\footnotesize{}5} & {\footnotesize{}8} & {\footnotesize{}9} & {\footnotesize{}10} & {\footnotesize{}14}\tabularnewline
{\footnotesize{}$\mathrm{NH_{3}}$} & {\footnotesize{}6} & {\footnotesize{}10} & {\footnotesize{}12} & {\footnotesize{}13} & {\footnotesize{}16}\tabularnewline
{\footnotesize{}$\mathrm{CH}_{4}$} & {\footnotesize{}7} & {\footnotesize{}12} & {\footnotesize{}14} & {\footnotesize{}16} & {\footnotesize{}18}\tabularnewline
{\footnotesize{}$\mathrm{O_{2}}$} & {\footnotesize{}9} & {\footnotesize{}10} & {\footnotesize{}11} & {\footnotesize{}13} & {\footnotesize{}20}\tabularnewline
{\footnotesize{}$\mathrm{H_{2}O_{2}}$} & {\footnotesize{}8} & {\footnotesize{}13} & {\footnotesize{}16} & {\footnotesize{}18} & {\footnotesize{}24}\tabularnewline
\bottomrule
\end{tabular*}
\par\end{centering}
\centering{}\caption{\justifying\emph{\small{}Qubit Requirement of Selected Molecules.}{\small{}
The qubit usage is calculated by taking the base $2$ logarithm of
the number of bases in \textcolor{red}{the STO-3G}\textcolor{red}{\sout{ a minimal}} basis set.}\label{table: example of qubit usage}}
\end{table}

The primary objective of this work is \textcolor{red}{to extend qubit-efficient encodings for these seniority-restricted spaces and} to evaluate the trade-off between quantum resource reduction and physical accuracy resulting from these restrictions \textcolor{red}{for Noisy Intermediate-Scale Quantum devices~\cite{Preskill2018}.} The results demonstrate that MHS achieves substantial qubit savings across a broad molecular test set and performs particularly well near equilibrium, where the restricted subspaces closely reproduce the reference energies with high fidelity. However, numerical benchmarks reveal a clear boundary of validity: in stretched-bond and open-shell regimes, the restricted pairing structure is insufficient to capture strong static correlation. Despite this limitation, the reduced subspaces offer significant practical advantages for variational quantum simulation by enhancing optimization stability and accuracy under shallow ansatz and moderate optimization budgets.

The remainder of this paper is organized as follows. Section 2 introduces the mathematical framework of restriction and encoding. Section 3 presents the generalized Hund-rule construction and the resulting subspace-efficiency estimation. Section 4 presents the numerical results, including qubit-reduction analysis, fidelity and potential energy surface benchmarks, and Variational Quantum Eigensolver (VQE) performance. We conclude in Section 5 with a discussion of the advantages and physical limitations of the proposed restricted-subspace framework.

\section{Mathematical Framework}

To reduce quantum resource requirements in molecular simulations, we restrict the molecular Hamiltonian to a physically motivated subspace before mapping it to qubits.

\subsection{Restriction and Extension}

Restricting the Hamiltonian to a target subspace reduces its complexity, but typically changes the original eigenvalues.

\begin{defn}
\label{Def: restriction-Reduction} Fix $\mathcal{F}$ a Hilbert space with an orthogonal basis \textcolor{red}{$ \{|\beta_{i}\rangle\}_i\in \mathcal{B}$}, and the target subspace \textcolor{red}{$\mathcal{F}^{'}:=\text{span}({|\beta_{i}\rangle}{_i\in\mathcal{B}^{'}})$} with \textcolor{red}{$\mathcal{B}^{'}\subset\mathcal{B}$} For a Hamiltonian $H\in\mathcal{L}(\mathcal{F})$ and the projector $P_{\mathcal{F}^{'}}:=\sum_{j\in\mathcal{B}^{'}}|\beta_{j}\rangle\langle\beta_{j}|$, the restricted Hamiltonian on $\mathcal{F}^{'}$ is the Hermitian operator $H|_{\mathcal{F}^{'}}:=(P_{\mathcal{F}^{'}})^{\dagger}HP_{\mathcal{F}^{'}}\in\mathcal{L}(\mathcal{F'})$.
\end{defn}

If the dimension of a restricted Hamiltonian does not match a power of two, one can extend the Hamiltonian to a larger space \textcolor{red}{\sout{while preserving the relevant eigenvalues.}}

\begin{defn}
\label{Def: Extension} Let $\mathcal{F}^{'}$ be a $k$-dimensional subspace of a $K$-dimensional Hilbert space $\mathcal{F}$. Let $V$ be a Hilbert space such that $\dim V+\dim\mathcal{F}^{'}=2^{n}$ for some $n\in\mathbb{N}$. Given $H\in\mathcal{L}(\mathcal{F})$, an extension of $H|_{\mathcal{F}^{'}}$ is a map $H|_{\mathcal{F}^{'}}\oplus H^{'}\in\mathcal{L}(\mathcal{F}^{'}\oplus V)$, where $H^{'}$ is Hermitian on $V$. If \textcolor{red}{$\dim V=2^{\lceil\log_{2}\dim\mathcal{F}^{'}\rceil}-\dim\mathcal{F}^{'}$}, then the extension is called a minimal extension of $H|_{\mathcal{F}^{'}}$. 
\end{defn}

Setting $H^{'}$ to a zero matrix \textcolor{red}{appends $\dim V$ zero eigenvalues to the spectrum of $H|_{\mathcal{F}'}$}\textcolor{red}{\sout{ strictly preserves the spectrum after extension}}\textcolor{red}{; the lowest eigenvalue of $H|_{\mathcal{F}'}$ is therefore preserved only when it is smaller than zero, which holds for every molecular ground state considered here}. Meanwhile, any linear encoding $\mathcal{E}$ with $ \mathcal{E^\dagger}\mathcal{E}=I$ is a linear change of coordinates and hence preserves the spectrum. Additionally, every Hamiltonian has a unique representation in the Pauli basis. Detailed proofs are provided in Appendix \ref{appendix:Proofs}.

\begin{restatable}{thm}{MyMainThm}
\label{Thm: encoding}With the same setup as above. Let $\mathcal{Q}$
be a $2^{n}$-dimensional Hilbert space with $\dim V+\dim\mathcal{F}^{'}=\dim\mathcal{Q}$
and $\mathcal{E}^{'}$ be a fixed encoding map from $\mathcal{F}^{'}\oplus V$
onto $\mathcal{Q}$. For any extension $H|_{\mathcal{F}^{'}}\oplus H^{'}\in\mathcal{L}(\mathcal{F}^{'}\oplus V)$,
the matrix $(H|_{\mathcal{F}^{'}}\oplus H^{'})^{\mathcal{\text{\ensuremath{\mathcal{E}^{'}}}}}:=\mathcal{E}^{'}\circ(H|_{\mathcal{F}^{'}}\oplus H^{'})\circ(\mathcal{E}^{'})^{\dagger}$
is Hermitian and preserves the spectrum of $H|_{\mathcal{F}^{'}}$.
\end{restatable}

In the following numerical tests, all calculations assume $H^{'}$ to be the zero matrix, with classical eigensolvers used by default. For VQE simulations, all Hamiltonians are decomposed into Pauli matrices.

\subsection{Generalized Hund's Rule and Dimension Estimation}

The original Hund's rule states that every orbital in a subshell is singly occupied before any orbital is doubly occupied, and all electrons in singly occupied orbitals share the same spin. Because an explicit atomic subshell classification is not generally available or convenient in molecular-orbital simulations, we introduce the following generalized filtering rule:

\begin{quote}[Generalized Hund's Rule]
\textit{Each orbital is filled with one spin-up electron before one spin-down electron}.
\end{quote}

By confirming that a state contains no lone spin-down electrons, we effectively filter the Fock space. The precise number of bases retained by this rule is given as follows:

\begin{defn}[Seniority and Hund subspaces]
\label{def:hund-subspaces}
\textcolor{red}{The \emph{seniority} $\hat\Omega(|\Phi\rangle)$ is the number of
singly occupied spatial orbitals of a given state $|\Phi\rangle$. A determinant obeys the
\emph{generalized Hund occupation rule} if every singly occupied orbital
contains an $\alpha$ spin electron; equivalently, no spatial orbital contains
a lone $\beta$ spin electron. The \emph{Hund subspace} $\mathrm{HS}$ is the
span of all $N$-electron determinants satisfying this rule. For an admissible seniority $\Omega \ge 0$, its seniority-$\Omega$ block is
\[
    \mathrm{MHS}(\Omega)
    :=
    \operatorname{span}
    \bigl\{
        \lvert\Phi\rangle\in\mathrm{HS}:
        \hat\Omega(|\Phi\rangle)=\Omega
    \bigr\},
\]
and consequently,
\[
    \mathrm{HS}
    =
    \bigoplus_{\Omega}\mathrm{MHS}(\Omega).
\]}
\end{defn}
\textcolor{red}{For a determinant in HS with seniority $\Omega$, the doubly
occupied orbitals have the same number of $\alpha$ and $\beta$ spin electrons, while each singly occupied orbital contains one $\alpha$ spin electron. Therefore,
\[
N_\alpha-N_\beta=\Omega,
\]
where $N_\alpha$ and $N_\beta$ denoted as total electron numbers for $\alpha$ and $\beta$ spin electrons.}
\textcolor{red}{
Every determinant in $\mathrm{MHS}(\Omega)$ thus has a fixed difference between $N_\alpha$ and $N_\beta$. In particular, $\mathrm{MHS}(0)$ is exactly the seniority-zero DOCI space. }

\textcolor{red}{With hypothesis~\textup{(H1)}, introduced below, the
entire $\mathrm{MHS}(\Omega)$ block is an eigenspace of $\widehat S^2$ with eigenvalues $S(S+1)$ where
$S=\Omega/2$. Assume further that $H$ conserves $N_\alpha$ and $N_\beta$ simultaneously. It follows that in such $H$, $\mathrm{MHS}(\Omega)\cap\mathrm{MHS}(\Omega')=\emptyset$ when  $\Omega\neq\Omega'$, and we reach the following reduction result.}

\begin{restatable}[Seniority-block reduction]{cor}{SeniorityBlock}
\label{cor:block}
\textcolor{red}{Suppose that $H$ commutes with $S_z := \frac{1}{2}( \widehat{N}_{\alpha} - \widehat{N}_{\beta})$ where $\widehat{N}_{\sigma}(|\Phi\rangle) = N_\sigma$ for any state $|\Phi\rangle$ and $\sigma= {\alpha, \beta}$. Then}
\textcolor{red}{\[
  P_{\mathrm{HS}}\cdot H\cdot P_{\mathrm{HS}}=\bigoplus_{\Omega} P_{\mathrm{MHS}(\Omega)}H\cdot P_{\mathrm{MHS}(\Omega)}.
\]}
\end{restatable}

\textcolor{red}{The Hund-space ground state is thus obtained by diagonalizing the individual blocks, $E_{\mathrm{HS}}=\min_{\Omega} E_{\mathrm{MHS}(\Omega)}$, and each $\mathrm{MHS}(\Omega)$ is encoded on $\lceil\log_2\dim\mathrm{MHS}(\Omega)\rceil$ qubits rather than the full $\mathrm{HS}$ on $\lceil\log_2\dim\mathrm{HS}\rceil$.}

\begin{restatable}[Spin purity under a common spatial basis]{cor}{SpinPurity}
\label{cor:spin}
\textcolor{red}{Denote $\{|\phi_p^\alpha\rangle\}_p$ and $\{|\phi_p^\beta\rangle\}_p$ as the spatial orbitals for $\alpha$ and $\beta$ spin electrons, respectively. If $\langle\phi_p^\alpha,\phi_q^\beta\rangle=\delta_{pq}$ (hypothesis (H1), as for an RHF reference). $\widehat S_+$ annihilates every determinant in $\mathrm{MHS}(\Omega)$, so $\mathrm{MHS}(\Omega)$ is an $\widehat S^2$-eigenspace of total spin $S=\Omega/2$. When (H1) fails (e.g. a symmetry-broken UHF reference), $\mathrm{MHS}(\Omega)$ is in general spin-contaminated.}
\end{restatable}

\textcolor{red}{
More generally, let
$\Sigma_{pq}:=\langle\phi_p^\alpha,\phi_q^\beta\rangle$ and, for a
Hund-compatible determinant $|\Phi\rangle$, let $\mathcal P$,
$\mathcal U$, and $\mathcal E$ denote its sets of doubly, singly
occupied, and empty spatial orbitals, respectively. Then
\[
\langle\Phi|\widehat S^2|\Phi\rangle
=
\frac{\Omega}{2}\left(\frac{\Omega}{2}+1\right)
+
|\mathcal P|
-
\sum_{p\in\mathcal P\cup\mathcal U}
\sum_{q\in\mathcal P}
|\Sigma_{pq}|^2.
\]
Thus, the nonnegative second term measures spin contamination caused by
the mismatch between the $\alpha$- and $\beta$-orbital bases. A proof is
provided in Appendix~\ref{appendix:Proofs}.}

\begin{restatable}{prop}{SizeEstimation}
\label{prop: SizeEstimation} For a molecule with $M$ orbitals and $N$ electrons, the number of fermionic bases\textcolor{red}{\sout{ in a minimal basis set}} satisfying the generalized Hund's rule is $\sum_{k=0}^{\lfloor\frac{N}{2}\rfloor}\binom{M}{N-k}\binom{N-k}{k}.$
\end{restatable}
To represent a target spin $S$, we take the unpaired electrons in the target high-spin configuration to have parallel spins~\cite{miessler1999inorganic}. Accordingly, we assume $N_{\alpha} - N_{\beta} = 2S$.
\begin{restatable}{cor}{MHS}
\label{Multi+hnud}
For a target spin $S$ represented by its highest-weight component
$S = \frac{1}{2}(N_{\alpha} - N_{\beta})$, the corresponding Hund block is $\mathrm{MHS}(\Omega)$ with
$\Omega=2S$, and its dimension is
\[
\dim\mathrm{MHS}(2S)
=
\binom{M}{\frac{N+2S}{2}}
\binom{\frac{N+2S}{2}}{\frac{N-2S}{2}}.
\]
\end{restatable}

Under hypothesis~\textup{(H1)}, this block has total spin $S$ and
multiplicity $2S+1$. For $S=0$, equivalently $\Omega=0$, the dimension reduces to
\[
\dim\mathrm{MHS}(0)=\binom{M}{N/2},
\]
which is the dimension of the seniority-zero DOCI
space~\cite{Bytautas2011}.


\begin{figure*}[t]
    \centering
    \includegraphics[width=\textwidth]{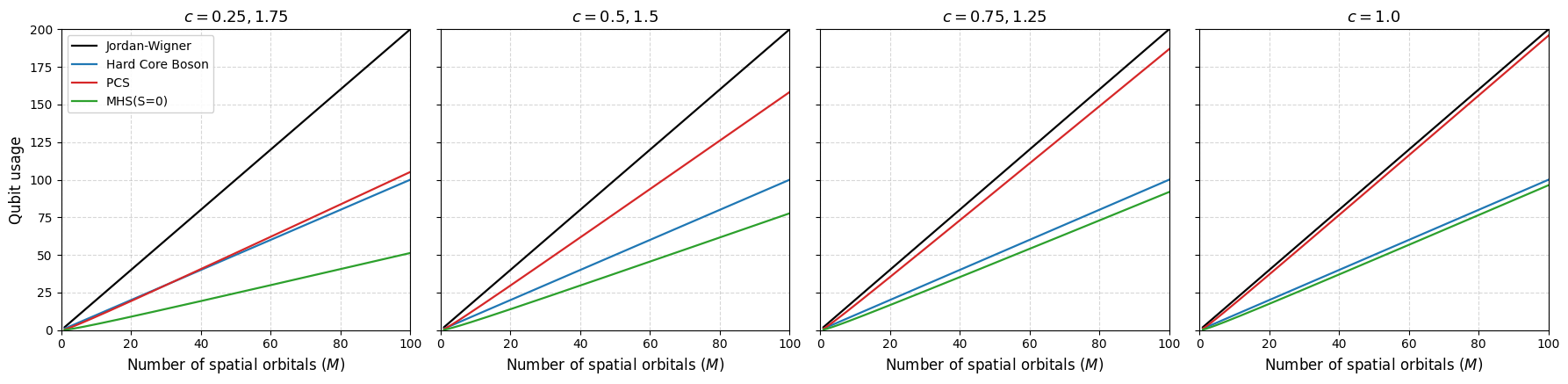}
    \justifying
    \caption{\emph{\small{Comparison of qubit requirements across different encoding strategies.} Jordan--Wigner, Hard Core Boson, PCS,  MHS(0) (MHS with no unpaired electrons, $S=0$), as a function of the number of spatial orbitals $M \in [0, 100]$ for fill factors $c:=N/M$.with $N$ denoting the number of electrons.}}
    \label{fig:qubit_usage_comparison}
\end{figure*}





\begin{restatable}[Resources estimation]{prop}{QubitSaving}
\label{prop:qubit-saving}
\textcolor{red}{
At a fixed $S$, as $M\to\infty$,}
\textcolor{red}{\[
  \bigl\lceil\log_2 \dim\mathrm{MHS}(\Omega)\bigr\rceil \;=\; M\cdot\mathrm{H}(c/2)+O(\log M) \]
  and
  \[\frac{\bigl\lceil\log_2 \dim\mathrm{MHS}(\Omega)\bigr\rceil }{\bigl\lceil\log_2 \dim \mathrm{PCS}\bigr\rceil}\rightarrow\frac{1}{2}\]
  where $\mathrm{H}(x):=-x\log_2x-(1-x)\log_2(1-x)$ is the binary entropy and $H(x)=H(1-x)$.}
\end{restatable}

\section{Implementation and Computational Efficiency}

The practical utility of the Multi-Hund Subspace (MHS) and Hund Subspace (HS) frameworks depends on whether the classical preprocessing overhead remains manageable relative to the quantum resource savings. In the NISQ regime, where qubit counts and gate depths are the primary constraints, the principal advantage of this strategy is the one-time classical initialization, which circumvents both classical memory limitations and quantum resource barriers.

The proposed simulation workflow incorporates a classical filtering process that precedes quantum circuit preparation:
\begin{enumerate}
\item \textbf{Mean-Field Initialization:} An initial Restricted Hartree-Fock (RHF) or Unrestricted Hartree-Fock (UHF) calculation is performed to obtain molecular orbitals, a standard process with computational scaling of $O(M^{3})$ to $O(M^{4})$ \cite{szabo2012modern, helgaker2014molecular}.
\item \textbf{Subspace Filtering:} The Generalized Hund’s Rule and molecular multiplicity constraints are applied to identify the relevant $D$-dimensional target subspace.
\item \textbf{Hamiltonian restriction:} An optimized mapping algorithm is used to project the Hamiltonian into the target basis. Then, transform the restricted Hamiltonian into a qubit-compatible format.
\end{enumerate}

\subsection{Restriction Complexity and Resource Estimation}

To ensure the viability of the method, it is necessary to evaluate the classical time required to construct the restricted matrix $H|_{\mathcal{F}'}$ and the resulting quantum measurement overhead.

\begin{figure*}[h]

\begin{tikzcd}[row sep=1.8cm, column sep=1cm]
\begin{array}{c} \text{Sorted Subspace Basis} \\ \{ |j\rangle \}_{j \in \Omega'} \end{array} &&&
\begin{array}{c} \text{Restricted Excitation} \\ E|_{\{ |j\rangle \}_{j \in \Omega'}} = \sum_{j \in \Gamma} c_j |\sigma(j)\rangle \langle j| \end{array} \\
\\
\begin{array}{c} \text{Permuted Basis} \\ \{ |\sigma(j)\rangle \}_{j \in \Omega'} \end{array} &&&
\begin{array}{c} \text{Remaining Basis} \\ \{ |\sigma(j)\rangle \}_{j \in \Gamma} \\ \Gamma := \{j \in \Omega' \mid \sigma(j) \in \Omega'\} \end{array}
\arrow["\begin{array}{c} \text{k-body Excitation,}\ E \\ \text{acts on basis} \\ E|j\rangle \rightarrow c_j\cdot|\sigma(j)\rangle \end{array}"', from=1-1, to=3-1]
\arrow[from=3-1, to=3-4, "{\begin{array}{c} \text{Check if} \\ |\sigma(j)\rangle \in \{ |j\rangle \}_{j \in \Omega'} \end{array}}"]
\arrow[from=3-4, to=1-4]
\end{tikzcd}
\centering{\caption{\justifying\emph{\small{}Workflow of Restriction.}{}}}
\end{figure*}

A standard molecular Hamiltonian contains $O(M^{4})$ terms. Evaluating the action of a creation or annihilation operator on a single computational basis state requires $O(1)$ time using bitwise operations. After generating the new state, verifying its presence in the target subspace and identifying its index requires $O(\log D)$ time via binary search. By exploiting the Hermiticity of the Hamiltonian, only the upper triangular portion is computed, reducing the workload by half. The total classical time complexity is therefore bounded by $O(M^{4} \cdot D \log D).$ \textcolor{red}{For MHS with fixed $N$ and $S$, this becomes $O\left(M^{4+\frac{N+2S}{2}}\log M\right)$.}

\textcolor{red}{The measurement overhead of VQE is characterized here by
the number of generalized Pauli terms $T$ and the encoded-Hamiltonian
one-norm $\Lambda$, which provides a conservative shot-budget $N_\text{s}$ bound at a
specified target precision. Since the electronic
Hamiltonian contains $O(M^4)$ fermionic terms, each determinant can be
connected to at most $O(M^4)$ other determinants. Consequently, the
number of nonzero entries in the upper triangle of the restricted matrix
is bounded by $O(DM^4)$. This bound governs the classical matrix-assembly
cost, but it does not bound the number of Pauli terms because a sparse
matrix can become dense in the generalized Pauli basis. In the present calculations, the molecular orbitals and the corresponding
one- and two-electron integrals are real. Consequently, the matrix
representation of the second-quantized Hamiltonian in the determinant
basis, as well as each restricted MHS block, is real symmetric. For a real
symmetric operator encoded on $q=\lceil\log_2D\rceil$ qubits, all Pauli
strings containing an odd number of $Y$ operators have zero coefficient.
Therefore,
\[
T\leq\frac{4^q+2^q}{2},
\]
as proved in Appendix~\ref{appendix:Proofs}. This is a structural upper bound rather than
an explicitly computed Pauli-term count; Among the twelve
$\mathrm{MHS}(\Omega)$ blocks reported in
Table~\ref{tab:resources}, the observed Pauli counts occupy between
$51.5\%$ and $100\%$ of the corresponding ceiling. These occupancies are
numerical observations and are not asserted as an asymptotic density
result.}

\begin{table}[htbp]
\scriptsize
\begin{tabular}{l l c c r c r r r}
\toprule
Molecule & Method & $\Omega$ & $S$ & $D$ & $q$ & $T$ & $\tfrac12(4^{q}{+}2^{q})$ & Occ.\ (\%) \\
\midrule
\multirow{4}{*}{$\mathrm{H_2}$}
  & MHS & 0 & $0$ & 2  & 1 & 3   & 3   & 100.0 \\
  & MHS & 2 & $1$ & 1  & 0 & 1   & 1   & 100.0 \\
  & HCB & 0 & $0$ & 2  & 2 & 6   & 10  & 60.0  \\
  & JW  & -- & -- & 6  & 4 & 15  & --  & --    \\ \addlinespace
\multirow{5}{*}{$\mathrm{LiH}$}
  & MHS & 0 & $0$ & 15 & 4 & 134  & 136  & 98.5  \\
  & MHS & 2 & $1$ & 60 & 6 & 2080 & 2080 & 100.0 \\
  & MHS & 4 & $2$ & 15 & 4 & 120  & 136  & 88.2  \\
  & HCB & 0 & $0$ & 15 & 6 & 52   & 2080 & 2.5   \\
  & JW  & -- & -- & 495 & 12 & 631 & --  & --    \\ \addlinespace
\multirow{4}{*}{$\mathrm{HF}$}
  & MHS & 0 & $0$ & 6  & 3 & 36   & 36   & 100.0 \\
  & MHS & 2 & $1$ & 15 & 4 & 120  & 136  & 88.2  \\
  & HCB & 0 & $0$ & 6  & 6 & 52   & 2080 & 2.5   \\
  & JW  & -- & -- & 66 & 12 & 631 & --  & --    \\ \addlinespace
\multirow{3}{*}{$\mathrm{OH}$}
  & MHS & 1 & $1/2$ & 30 & 5 & 504 & 528 & 95.5 \\
  & MHS & 3 & $3/2$ & 20 & 5 & 272 & 528 & 51.5 \\
  & JW  & -- & -- & 220 & 12 & 631 & --  & --   \\ \addlinespace
\multirow{4}{*}{$\mathrm{CH}$}
  & MHS & 1 & $1/2$ & 60 & 6 & 1872 & 2080 & 90.0 \\
  & MHS & 3 & $3/2$ & 60 & 6 & 2072 & 2080 & 99.6 \\
  & MHS & 5 & $5/2$ & 6  & 3 & 32   & 36   & 88.9 \\
  & JW  & -- & -- & 792 & 12 & 631 & --  & --   \\
  \addlinespace
\bottomrule
\end{tabular}
\centering
\setlength{\tabcolsep}{5pt}
\caption{\textcolor{red}{\justifying\emph{\small{}Resource counts across encodings (STO-3G).}{\small{} $D$: represented dimension (MHS/HCB block dimension; for JW the particle-conserving dimension $\binom{2M}{N}$); $q$: qubits; $T$: number of Pauli terms (identity included); ceiling $\tfrac12(4^q+2^q)$ and occupancy $=T/\text{ceiling}$ apply to the compressed real-symmetric encodings (for JW on $2M$ qubits the ceiling is $\sim\!8.4\times10^{6}$, occupancy $<0.01\%$, omitted). MHS and JW use generalized/standard Pauli operators; HCB uses weight-$\le2$ standard Paulis.}}}
\label{tab:resources}
\end{table}


\textcolor{red}{Within the seniority-zero sector, $\mathrm{MHS}(0)$ and
the hard-core-boson (HCB) encoding span the same DOCI space but produce
different Hamiltonian representations. The compact MHS encoding uses
fewer qubits, whereas the HCB encoding preserves an orbital-local
structure with Pauli weight at most two. We characterize the associated
measurement cost using
$\Lambda=\sum_{P\neq I}|c_P|.$
Under optimal measurements, the number of shots $N_{\mathrm{shots}}$ required to achieve precision $\epsilon$ is
bounded by \[N_{\mathrm{shots}}\leq\frac{\Lambda^2}{\epsilon^2}\]\cite{Rubin2018}. The computed
values of $\Lambda$ and the resulting fixed-precision shot-bound ratios
are reported in Table~\ref{tab:hcb_resources}.}

\begin{table}[h]
\centering
\begin{tabular}{lrrrrrrr}
\toprule
& \multicolumn{2}{c}{$q$}
& \multicolumn{2}{c}{$T$}
& \multicolumn{2}{c}{$\Lambda$}
& Shot-bound ratio \\
\cmidrule(lr){2-3}
\cmidrule(lr){4-5}
\cmidrule(lr){6-7}
Molecule
& MHS & HCB
& MHS & HCB
& MHS & HCB
& $\Lambda_{\mathrm{MHS}}^2/\Lambda_{\mathrm{HCB}}^2$ \\
\midrule
$\mathrm{H_2}$   & 1 & 2 & 3    & 6  & 0.969   & 1.542  & 0.40 \\
$\mathrm{LiH}$   & 4 & 6 & 134  & 52 & 9.545   & 9.302  & 1.05 \\
$\mathrm{HF}$    & 3 & 6 & 36   & 52 & 98.638  & 64.707 & 2.32 \\
$\mathrm{BeH_2}$ & 6 & 7 & 1612 & 71 & 39.044  & 16.112 & 5.87 \\
$\mathrm{H_2O}$  & 5 & 7 & 504  & 71 & 154.122 & 54.790 & 7.91 \\
\bottomrule
\end{tabular}
\caption{\emph{Resource comparison between the compact MHS and
hard-core-boson encodings of the seniority-zero space.}
Both encodings represent the same DOCI space. Here, $q$ is the number of
qubits, $T$ is the number of Pauli terms including the identity, and
$\Lambda=\sum_{P\neq I}|c_P|$. The final column gives the ratio of the
shot-budget bounds $\Lambda^2/\epsilon^2$ at fixed target precision
$\epsilon$.}
\label{tab:hcb_resources}
\end{table}

\textcolor{red}{At the chemical-accuracy scale
$\epsilon=1.6\times10^{-3}$~Ha, the corresponding conservative
shot-budget bounds $(N_{\mathrm{MHS}},N_{\mathrm{HCB}})$ are
$(3.67\times10^5,9.29\times10^5)$ for $\mathrm{H_2}$,
$(3.56\times10^7,3.38\times10^7)$ for LiH,
$(3.80\times10^9,1.64\times10^9)$ for HF,
$(5.96\times10^8,1.01\times10^8)$ for $\mathrm{BeH_2}$, and
$(9.28\times10^9,1.17\times10^9)$ for $\mathrm{H_2O}$. These coefficient-based bounds should not be interpreted as device-specific shot predictions. The practical measurement cost depends on the prepared state and may be reduced through commuting-term
grouping, covariance-aware allocation, and related measurement
strategies~\cite{Crawford2021,Yen2023Measurement}.}


\textcolor{red}{Although the compact MHS encoding uses fewer qubits than
HCB in every case considered, it does not provide a uniform advantage in
Pauli-term count or measurement cost. MHS has fewer Pauli terms for
$\mathrm{H_2}$ and HF, but more for LiH, $\mathrm{BeH_2}$, and
$\mathrm{H_2O}$. Similarly, the $\Lambda$-based shot-bound ratio is
$0.40$ for $\mathrm{H_2}$ but exceeds unity for the other four systems,
increasing to $7.91$ for $\mathrm{H_2O}$. Qubit count, Pauli-term count,
Pauli weight, and measurement cost must therefore be treated as distinct
resource metrics.}

\textcolor{red}{
The choice of baseline is critical when evaluating the efficiency of
this framework:}
\begin{itemize}
\item \textbf{Versus standard JW encoding:}\textcolor{red}{
The compact MHS encoding reduces the qubit usage from $2M$ to
$q=\lceil\log_2D\rceil$. However, this reduction does not imply a corresponding
reduction in Pauli terms, because the compressed Hamiltonian can be dense
in the generalized Pauli basis.}

\item \textbf{Versus PCS and MS encodings:}\textcolor{red}{
For a fixed target sector, MHS has the smallest represented dimension
among these restricted spaces and therefore requires no more qubits.}

\item \textbf{Versus HCB encoding:}\textcolor{red}{
At $\Omega=0$, MHS and HCB encode the same DOCI space. MHS minimizes the
qubit usage, whereas HCB generally produces fewer and lower-weight
Pauli terms. Their relative advantage therefore depends on whether qubit
count or measurement cost is the binding resource.}
\end{itemize}

\textcolor{red}{
In the $\Omega > 0 $ case, HCB cannot be applied directly. Other approaches to nonzero-seniority sectors
have been developed~\cite{HendersonChenScuseria2026}; however, a direct resource comparison with these methods is outside the scope of this study.}

\textcolor{red}{The effect of the restriction on classical memory is illustrated using a linear $\mathrm{H}_{22}$ chain in the STO-3G basis, which consists of $44$ spin orbitals and $22$ electrons.
A particle-conserving statevector contains $\binom{44}{22}\approx2.1\times10^{12}$ non-zero entries. A real-valued statevector requires approximately $8.42$~TB in single precision or
$16.84$~TB in double precision, whereas a general \texttt{complex128} statevector requires approximately $33.7$~TB. Because the Hamiltonian matrices considered here are real symmetric,
our classical eigensolver uses real \texttt{float64} arithmetic. In comparison, the selected MHS block has dimension $705{,}432$. The generic $O(DM^4)$ nonzero bound is not tight in the seniority-zero
sector: only diagonal terms and single-pair hops connect determinants. Consequently, the upper-triangular nonzero Hamiltonian values require approximately $0.35$~GB in double precision, excluding sparse indices. The implementation assembles these values in batches and does not construct the complete PCS statevector.
}

The implementation of the Multi-Hund Subspace (MHS) restriction provides a significant reduction in the dimensionality of the Hamiltonian (Table \ref{tab:systemsize}).  This reduction enhances the feasibility and effectiveness of numerical tests, enabling analysis of complex systems. For instance, \textcolor{red}{the lowest energy within the selected $\mathrm{H}_{22}$ MHS(0) block, $-11.48939$ Hartree, can be computed classically} in 30 minutes on an AWS g4dn.8xlarge instance (32 threads, 128 GB RAM, 16 GB VRAM) \cite{braket}. See Table \ref{tab:H22}.

Although SRS extends the reach of classical simulation, it too faces a memory limitation as system size grows. \textcolor{red}{\sout{The primary paradigm shift offered by the MHS framework is realized through its application to quantum hardware.}} \textcolor{red}{On the quantum side,} standard JW encoding for $H_{22}$ requires 44 qubits\textcolor{red}{, whereas the selected MHS(0) block requires $20$ qubits (Table~\ref{tab:systemsize}). This comparison demonstrates a reduction in qubit requirements; however, it does not establish the feasibility of the corresponding quantum calculation, as circuit and measurement costs must also be considered. The $\mathrm{H}_{22}$ result reported here was obtained using a classical eigensolver and does not represent a quantum-hardware demonstration.}

\section{Results and Discussion}

We analyze the behaviour of the Multi-Hund Subspace (MHS) to determine its suitability as a resource-efficient molecular simulation method on quantum computers. Our discussion focuses on three key areas: stability against reference selection, the physical limitations revealed by potential energy surface (PES) scans, and the practical payoff in facilitating Variational Quantum Eigensolver (VQE) convergence. All following results are computed with the STO-3G basis set, and molecular geometries are from the Computational Chemistry Comparison and Benchmark Database \cite{CCCBDB}.

\subsection{Reference Sensitivity and the Cost-Accuracy Trade-off}
\label{sec:reference}
\textcolor{red}{The energy cost introduced by the subspace restriction and its sensitivity to the choice of mean-field orbitals are first examined. Let $E_{\mathrm{FCI}}$ denote the ground-state energy in PCS and $E_{\mathrm{MHS}(\Omega)}$ denote the exact ground-state energy within $\mathrm{MHS}(\Omega)$. The irreducible subspace error is defined as \[ \Delta E(\Omega) := E_{\mathrm{MHS}(\Omega)}-E_{\mathrm{FCI}} \geq 0, \] where the inequality follows from the variational principle. If a VQE run returns an energy $E_r$, then \[ E_r-E_{\mathrm{FCI}} = \underbrace{E_r-E_{\mathrm{MHS}(\Omega)}}_{\epsilon_{\mathrm{opt}}} + \underbrace{E_{\mathrm{MHS}(\Omega)}-E_{\mathrm{FCI}}}_{\Delta E(\Omega)}. \] Here, $\epsilon_{\mathrm{opt}}$ quantifies incomplete optimization, whereas $\Delta E(\Omega)$ is imposed by the subspace restriction and cannot be reduced by increasing the circuit depth or changing the optimizer. Figure~\ref{fig:cost_accuracy} reports $\Delta E(\Omega)$ for every seniority block in the molecular test set.}


\begin{figure}[h]
    \centering
    \includegraphics[width=0.9\textwidth]{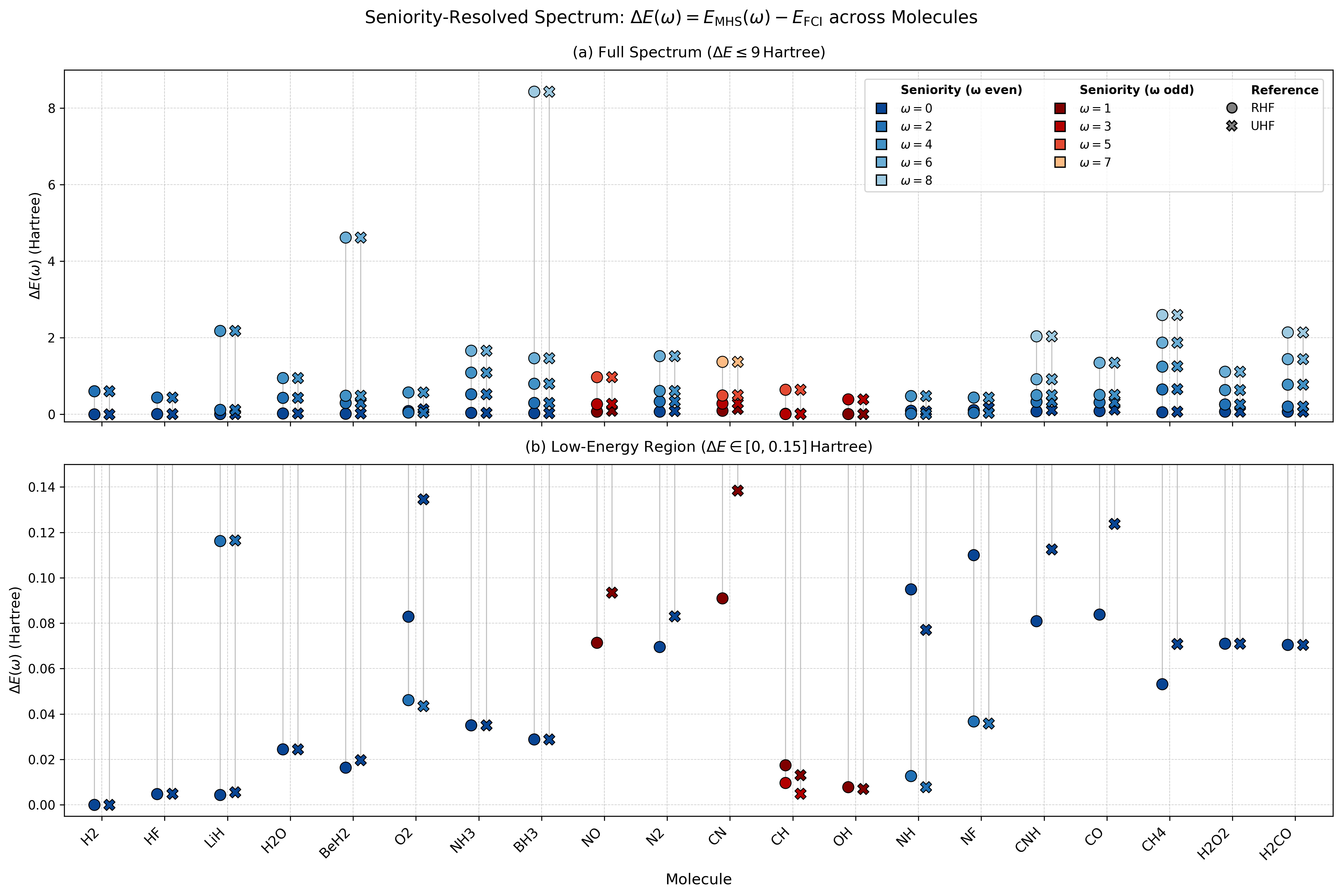}
    \caption{\textcolor{red}{\justifying\emph{\small{}Seniority-resolved energy cost of the subspace restriction.}{\small{} For each molecule, $\Delta E(\Omega)=E_{\mathrm{MHS}(\Omega)}-E_{\mathrm{FCI}}\ge 0$ is shown for every seniority block $\Omega$ (blue: even; red: odd), under the RHF (circles) and UHF (crosses) references. (a)~full range; (b)~the low-energy region. The underlying data are listed in Appendix~\ref{appendix:Data}.}}}
    \label{fig:cost_accuracy}
\end{figure}

\begin{figure}[h]
    \centering
    \includegraphics[width=0.9\textwidth]{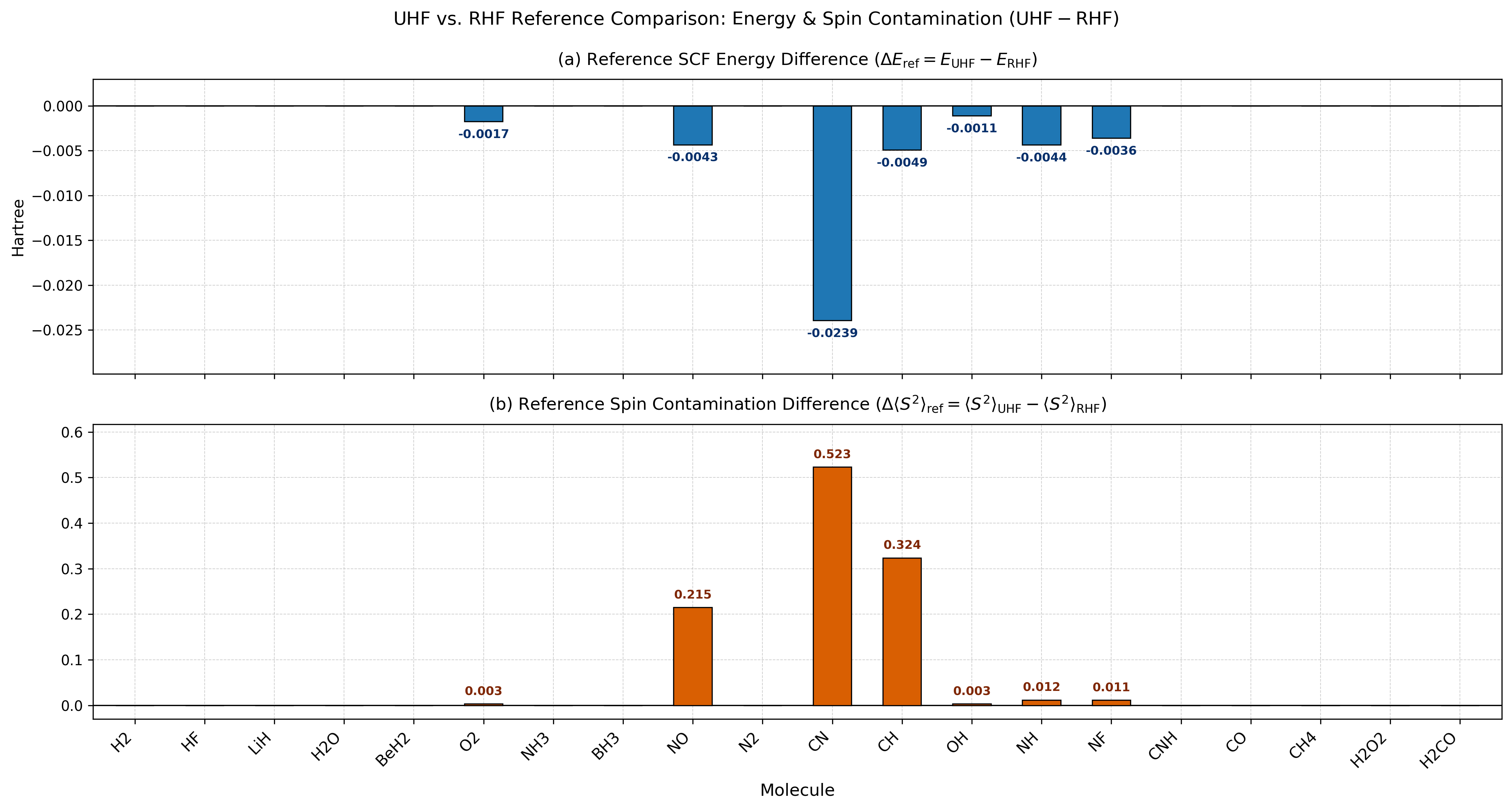}
    \caption{\textcolor{red}{\justifying\emph{\small{}Reference dependence of the two mean-field starting points.}{\small{} (a)~SCF energy difference $\Delta E_{\mathrm{ref}}=E_{\mathrm{UHF}}-E_{\mathrm{RHF}}$; (b)~reference spin contamination $\Delta\langle S^2\rangle_{\mathrm{ref}}=\langle S^2\rangle_{\mathrm{UHF}}-\langle S^2\rangle_{\mathrm{RHF}}$, across the molecular test set. Only the open-shell radicals (e.g.\ CN, CH, NO) show appreciable contamination; the closed-shell systems have $\Delta\langle S^2\rangle_{\mathrm{ref}}=0$.}}}
    \label{fig:uhf_vs_rhf}
\end{figure}

\textcolor{red}{\sout{A striking feature observed in our data is the diversity in solution variance. For the Multiplicity Subspace (MS) and Particle Conservation Subspace (PCS), we observe significant vertical scatter for the same molecule depending on the reference orbitals. In sharp contrast, the MHS and Hund Subspace (HS) data points appear tightly clustered, exhibiting minimal diversity between RHF and UHF calculations. This suggests that by enforcing strict generalized Hund's rule constraints---specifically pairing and local spin maximization---MHS and HS capture the same essential physical characters regardless of initial symmetry breaking in the reference.}}

\textcolor{red}{The MHS and HS constructions are themselves reference-dependent in a manner not apparent from the mean-field energy. The subspaces constructed from RHF and UHF orbitals are distinct, non-nested subspaces within the same particle-conserving space, resulting in unordered energies. Across the $71$ MHS test set, the UHF minus RHF change in $E_{\mathrm{MHS}}$ ranges from $-64$ to $+18$ mHa, with 26 of 71 MHS tests exhibiting higher energy in UHF reference case. Notably, six closed-shell molecules—CO, CNH, CH$_4$, N$_2$, BeH$_2$, and LiH—have RHF and UHF references with numerically identical SCF energies ($\Delta E_{\mathrm{SCF}}\le 2\times10^{-11}$ Ha), yet their MHS energies differ by up to 40 mHa (CO 39.96, CNH 31.48, CH$_4$ 17.66, N$_2$ 13.45, BeH$_2$ 3.30, LiH 1.23 mHa). }

\textcolor{red}{This observation demonstrates that The MHS energy depends on the spatial-orbital basis used to define seniority, not merely on the SCF energy or reference determinant. As established in pairing- and seniority-based methods~\cite{Bytautas2011,Stein2014}, the determinant subspace is not invariant under orbital rotations. Although the underlying reference determinant and its SCF energy remain unchanged, rotations within a degenerate space, such as the $\pi$ orbitals of CO, CNH, N$_2$, BeH$_2$ or the $t_2$ orbitals of CH$_4$, modify the determinant set that defines the paired subspace. Therefore, it is therefore necessary to specify the spatial orbital basis to ensure the reproduciblility of MHS.~\cite{Bytautas2011,Stein2014}. In contrast, This reference choice does not affect the spin quality of the blocks themselves: under an RHF reference every $\mathrm{MHS}(\Omega)$ block is an exact eigenstate of $\widehat S^{2}$ with $S=\Omega/2$, whereas under a UHF reference the blocks may be spin-contaminated (Appendix~\ref{appendix:Data}).}

\textcolor{red}{\sout{However, this stability}}\textcolor{red}{Meanwhile, the high qubit efficiency of MHS} requires a deliberate trade-off. As shown in Figure \ref{fig:cost_accuracy}, MHS and HS generally do not achieve chemical accuracy, defined as an error below $1.6\times10^{-3}$ Ha, except $H_2$ in this test set. This standard is often achievable with the more computationally expensive MS. This represents a strategic compromise: \textcolor{red}{The irreducible subspace error of MHS, $\Delta E(\Omega)$, enables a substantial reduction in dimensionality-approximately $5$ to $100$} smaller than MS. In the context of NISQ devices, where qubit usage are the primary bottleneck, this trade-off is often essential to make simulation feasible.

\textcolor{red}{To characterize which part of the exact ground state each block captures, we compute the subspace fidelity $F=\lVert P_{\mathrm{deg}}|\psi_{\mathrm{MHS} (\Omega)}\rangle\rVert^{2}$ where
$P_{\mathrm{deg}}:=\sum_{k=1}^{g}
\left|\psi_{\mathrm{PCS},k}\right\rangle
\left\langle\psi_{\mathrm{PCS},k}\right|$
and $\left|\psi_{\mathrm{PCS},k}\right\rangle$ is the $ k$th-degenerate ground state, the weight of the block ground state on the (possibly degenerate) PCS ground-state eigenspace. Figure~\ref{fig:fidelity} reports $F$ for every seniority block of the molecular test set.}

\begin{figure}[h]
    \centering
    \includegraphics[width=0.9\textwidth]{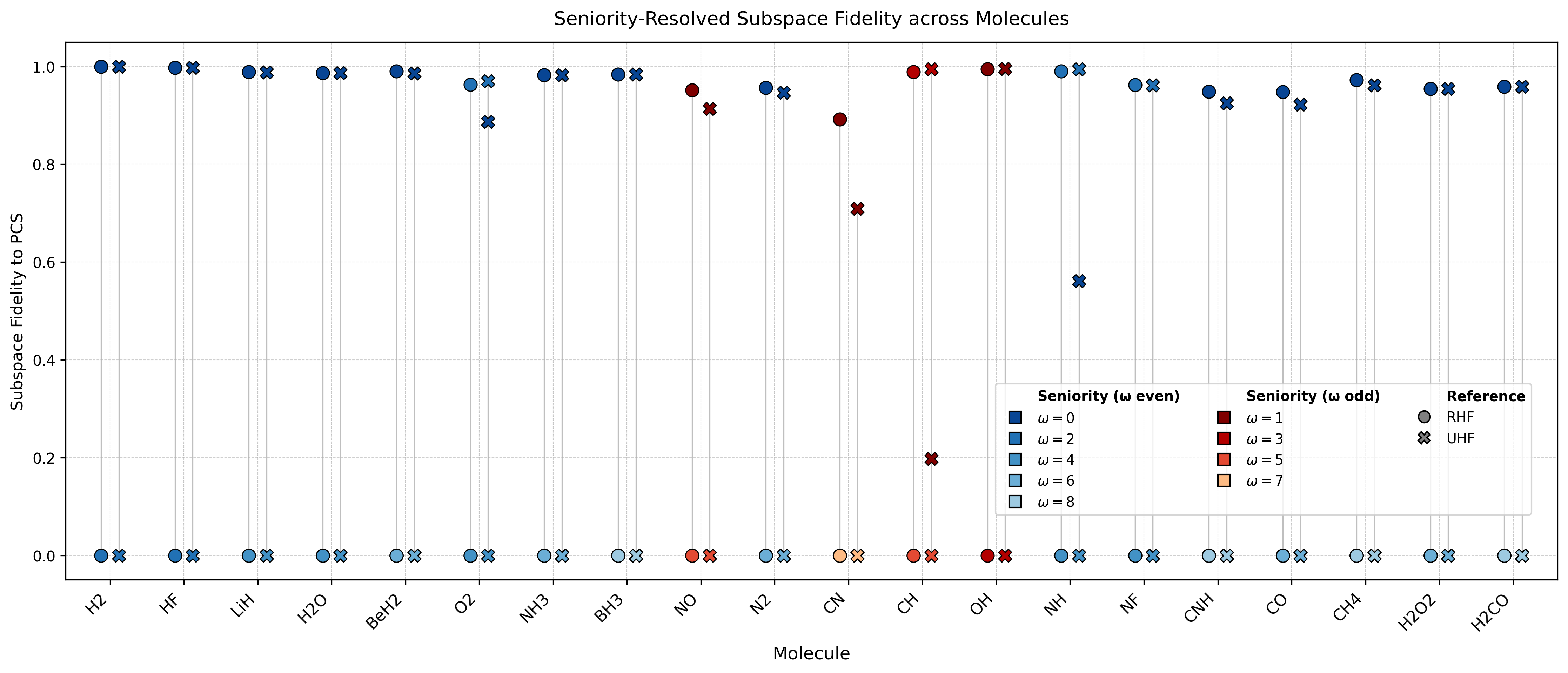}
    \caption{\textcolor{red}{\justifying\emph{\small{}Seniority-resolved subspace fidelity to the exact ground state.}{\small{} For each molecule, $F=\lVert P_{\mathrm{deg}}\cdot\psi_{\mathrm{MHS}(\Omega)}\rVert^{2}$ is shown for every block $\Omega$ (blue: even; red: odd), under the RHF (circles) and UHF (crosses) references. In each molecule a single block dominates the overlap. Under RHF it is the only block with nonzero weight---any block whose spin $S=\Omega/2$ differs from the ground-state spin is exactly orthogonal ($F=0$)---whereas under UHF this exactness is lost: the matched block can lose weight (e.g.\ CN, NH) and mismatched blocks can gain nonzero overlap (e.g.\ the $\Omega=3$ block of CH). The vertical axis is linear.}}}
    \label{fig:fidelity}
\end{figure}

\textcolor{red}{The restriction serves as a strict spin-selective filter. By definition, MHS blocks associated in different seniority sectors $\Omega$ are mutually orthogonal. Crucially, under RHF, each block is spin-pure with total spin $S = \Omega/2$ (Appendix~\ref{appendix:Data}), so distinct blocks correspond to distinct total spin sectors. Since the exact ground state $\vert{}\Psi_0\rangle$ has a definite spin $S_g$, any block with $S \neq S_g$ remains strictly orthogonal to $\vert{}\Psi_0\rangle$. Therefore, projection onto the exact ground state yields non-zero fidelity for only one block, specifically the block where $S = \Omega/2 = S_g$, while all other blocks have strictly zero overlap. Under a UHF reference, however, spin contamination breaks this spin purity and eliminates the clean single-block selection. As a result, the fidelity of spin-matched block loses and falls below the RHF baseline, as observed in CN and NH. Meanwhile, spin-mismatched blocks that are strictly zero under RHF acquire non-zero overlap. For example, the $\Omega = 3$ block of CH increases from 0 to approximately 0.2.}


\subsection{Potential Energy Surface Analysis}
\textcolor{red}{The physical limitations of the restricted subspaces are characterized by computing the potential energy curves of $\mathrm{N}_{2}$ and $\mathrm{LiH}$ along their dissociation coordinates. Near dissociation, strong static correlation leads to near-degeneracies, making restricted-subspace energies highly sensitive to orbital rotations among degenerate states. Although these unitary transformations leave the SCF energy invariant, they alter the determinant basis defining MHS, causing unconstrained scans to show spurious discontinuities. To ensure smooth, well-defined curves, this rotational freedom is resolved sequentially along the coordinate. First, the SCF reference is propagated bidirectionally between adjacent bond lengths, with the lowest-energy converged solution at each geometry. Second, degenerate orbitals between neighbouring geometries are aligned using a unitary rotation that maximizes orbital overlap, thereby preserving all SCF observables and maintaining invariance of the full configuration interaction (FCI) energy. Each restricted energy is subsequently evaluated as the lowest eigenvalue of the Hamiltonian within the corresponding subspace, and fidelity is defined as the projection of this state onto the (possibly degenerate) FCI ground-state eigenspace.}

\begin{figure}[h]
    \centering
    \includegraphics[width=\textwidth]{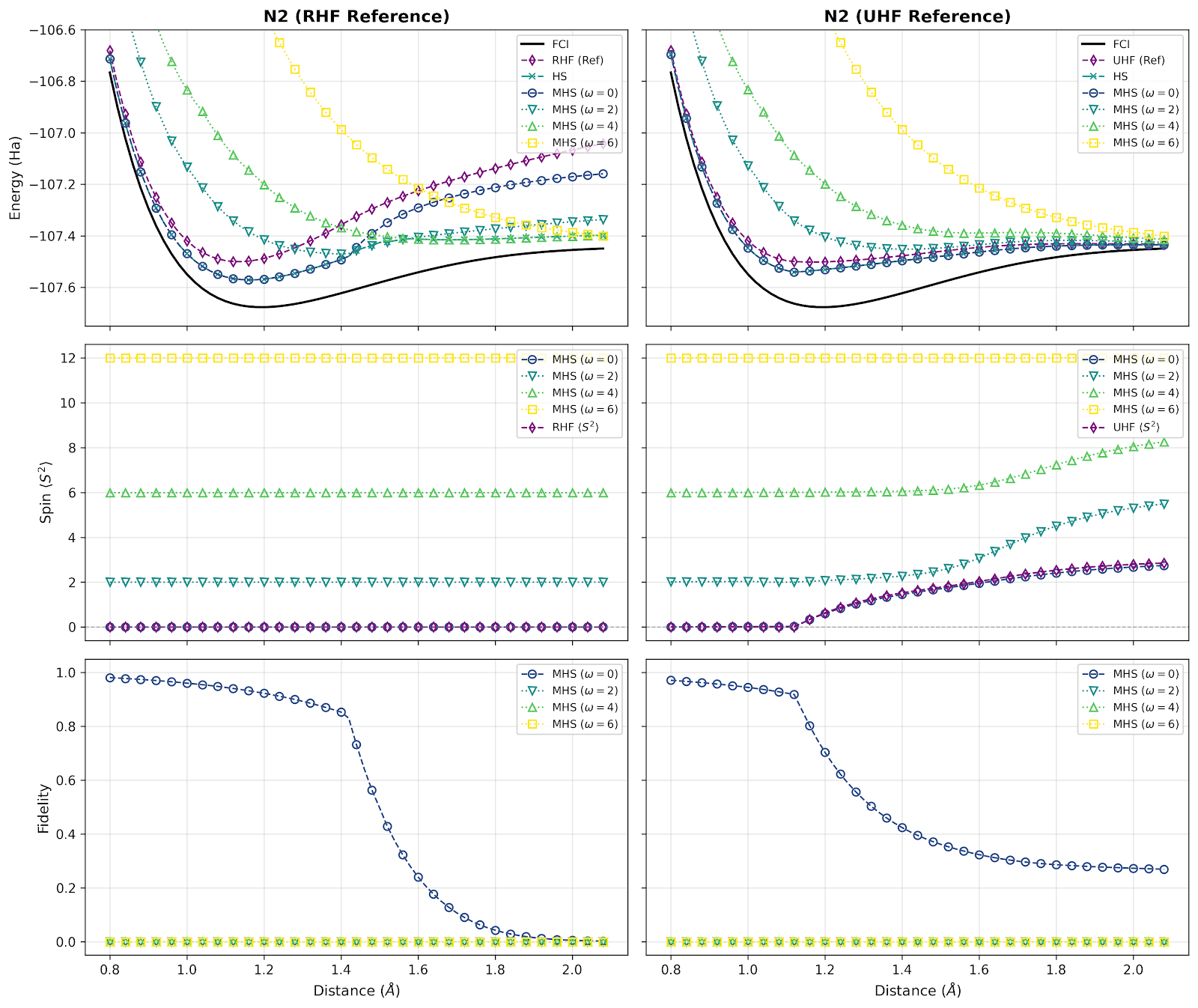}
    \caption{\textcolor{red}{\textit{\small{}Potential Energy Surface of $\mathrm{N}_{2}$.} \small{}Energy (top row), spin $\langle S^2\rangle$ (middle row), and fidelity (bottom row) as functions of bond length, for the RHF reference (left column) and the UHF reference (right column). Each panel compares FCI, the mean-field reference, the Hund Subspace (HS), and the seniority-resolved MHS$(\Omega)$ blocks.}}
    \label{fig:n2_pes}
\end{figure}

\begin{figure}[h]
    \centering
    \includegraphics[width=\textwidth]{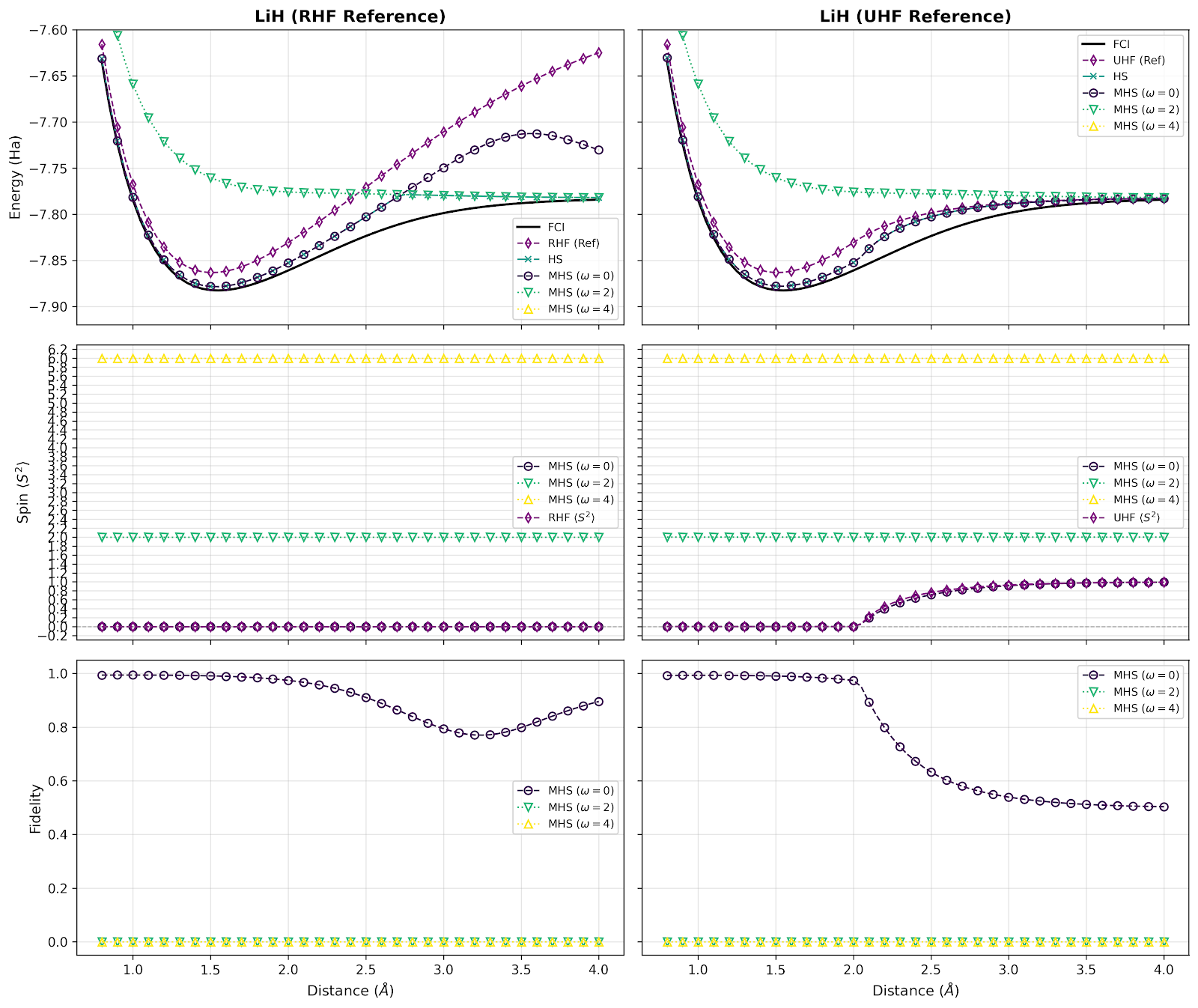}
    \caption{\textcolor{red}{\textit{\small{}Potential Energy Surface of LiH.} \small{}Energy (top row), spin $\langle S^2\rangle$ (middle row), and fidelity (bottom row) as functions of bond length, for the RHF reference (left column) and the UHF reference (right column). Each panel compares FCI, the mean-field reference, the Hund Subspace (HS), and the seniority-resolved MHS$(\Omega)$ blocks.}}
    \label{fig:lih_pes}
\end{figure}

In the near-equilibrium regime ($R < 1.1$ \AA\ for $\mathrm{N}_{2}$ and $R < 2.0$ \AA\ for LiH), the MHS and HS energies closely follow the FCI curve and maintain high fidelity. In this region, before the onset of spin contamination ($\langle S^{2} \rangle = 0$), the RHF- and UHF-based results are nearly identical. This indicates that, when dynamic correlation \cite{helgaker2014molecular} dominates, the generalized Hund-type constraints capture the essential low-energy physics largely, independently of the reference choice. For NISQ applications focused on equilibrium properties, MHS therefore provides a reliable and substantial reduction of the Fock-space dimension.

As the bond is stretched, however, the main limitation of these restricted subspaces becomes apparent: they cannot fully capture static correlation \cite{helgaker2014molecular}. In the dissociation process, especially for the triple bond in $\mathrm{N}_{2}$, the exact ground state includes pronounced multi-reference character, involving superpositions of nearly degenerate electronic configurations. Because the MHS and HS constructions impose generalized Hund-type pairing constraints, they lack the flexibility required to represent the dissociation-limit wavefunction. Consequently, their energies deviate from the FCI curve, and their fidelities drop.

\textcolor{red}{The spin-selective structure identified in the previous section is also evident in the dissociation curves, with its breakdown corresponding to the Coulson–Fischer transition. Under the RHF reference, each seniority block maintains its exact spin $S=\Omega/2$ at all bond lengths. As a result, the blocks remain spin-distinct, and the lowest-energy block is unique at each molecular structure. As the bond stretches, this ground block transfers to successively higher seniority: once for LiH ($\Omega=0\to2$ near $R\approx2.7$~\AA) and in a stepwise manner for $\mathrm{N}_2$ ($\Omega=0\to2\to4\to6$ near $R\approx1.43,\ 1.57,\ 2.07$~\AA). In contrast, under the UHF reference, this clean transfer does not occur. At the Coulson–Fischer point~\cite{CF_point, Coulson01041949}, the molecular orbitals transition from delocalized bonding orbitals to localized broken-symmetry orbitals. The $\Omega=0$ block ceases to be spin-pure, with its $\langle S^2\rangle$ departing from zero near $R\approx1.14$~\AA\ for $\mathrm{N}_2$ (reaching $2.74$ at $R=2.08$~\AA) and near $R\approx2.1$~\AA\ for LiH (reaching $0.99$ at $R=4.0$~\AA), and it remains the lowest block throughout. The higher-spin character that RHF assigns to a distinct higher-seniority block is, under UHF, absorbed into a symmetry-broken $\Omega=0$ block. Although pairing can be formulated in a broken-symmetry basis~\cite{Bajdich2006,Scuseria2011}, the present MHS construction imposes a fixed same-index pairing restriction without optimizing the orbitals or pairing structure. The observed loss of accuracy therefore characterizes this particular restricted construction rather than pairing models in general. Consequently, the fidelities drop sharply, the energies diverge from FCI, and the restricted-pairing picture is no longer quantitatively accurate.}

Figure \ref{fig:n2_pes} and \ref{fig:lih_pes} also demonstrate that the long-range behaviour of MHS and HS depends on the orbital reference used to define the subspace. Near equilibrium, RHF- and UHF-based constructions give essentially the same results, while after symmetry breaking, UHF-based restricted subspaces deteriorate much more rapidly. Thus, while the low-energy physics is captured similarly in the weakly correlated regime, the stretched-bond regime indicates a pronounced reference dependence in the restricted-subspace method.

Finally, the fidelity curves do not decay to zero in the dissociation regime, but instead approach nonzero plateaus. In the present scans, these asymptotic values depend on both the molecule and the orbital reference. \textcolor{red}{For example, the UHF-based scans approach approximately $0.5$ for LiH and $0.25$ for $\mathrm{N}_{2}$.\sout{ whereas the RHF-based scans remain noticeably higher.} This asymptotic value is reference-dependent, but not in a uniform direction: for LiH the RHF-based scan stays higher (above $\sim0.8$), whereas for $\mathrm{N}_{2}$ the RHF-based fidelity falls to near zero, below its UHF counterpart.} \textcolor{red}{This behaviour is consistent with the interpretation that the restricted subspaces retain only a subset of the degenerate configurations contributing to the exact dissociation-limit.} By excluding certain uncoupled spin configurations, MHS and HS capture only part of the exact wavefunction, which naturally leads to nonzero fidelity plateaus.

\subsection{VQE Performance and Convergence}

\textcolor{red}{The noiseless VQE benchmark for the restricted subspaces employs the hardware-efficient ansatz shown in Figure~\ref{fig:hea} at depths $L=1$ to $7$, utilizing both Adam and L-BFGS optimizers with $15$ random ansatz-parameter initializations per configuration. For each $\mathrm{MHS}(\Omega)$ block of HF, LiH, OH, and CH, box plots with scatter points are presented. All numerical results are included in Figures~\ref{fig:vqe_2block}--\ref{fig:vqe_3block}, except for runs that reached the $1000$-iteration budget without converging ($27$ runs, all with Adam, in the OH and CH ground blocks). For comparison, in Figure ~\ref{fig:uccsd},the UCCSD ansatz is also evaluated under the Jordan--Wigner encoding on the full active space.}

\begin{figure}[h]
\centering
\resizebox{0.4\textwidth}{!}{%
  \text{U$(\mathbb{\boldsymbol{\theta}})$ =}
  \begin{quantikz}
  & \qw& \gate{Ry(\theta)} \gategroup[wires=6,steps=5,style={dotted,cap=round,inner sep=10pt}]{} & \ctrl{1} & \qw & \qw & \qw & \qw  &\gate{Ry(\theta)}& \qw\\
  & \qw & \gate{Ry(\theta)} & \targ{} &\ctrl{1} & \qw & \qw & \qw & \gate{Ry(\theta)}& \qw\\
  & \qw & \gate{Ry(\theta)} & \qw & \targ{} &\ctrl{1} & \qw & \qw &  \gate{Ry(\theta)}& \qw\\
  & \qw & \gate{Ry(\theta)} & \qw & \qw & \targ{} &\ctrl{1} & \qw &  \gate{Ry(\theta)}& \qw\\
  & \qw & \gate{Ry(\theta)} & \qw  & \qw & \qw & \targ{} & \qw & \gate{Ry(\theta)}& \qw\\
  &  &  & & \vdots &   &  &  \times L &  &
  \end{quantikz}
} 
\caption{\emph{\small{}Hardware efficiency Ansatz. \cite{HardwareEfficientAnsatz, kandala2017hardware}}\small{} \emph{L}: the number of layers}
\label{fig:hea}
\end{figure}

\begin{figure}[h]
    \centering
    \includegraphics[width=0.7\textwidth]{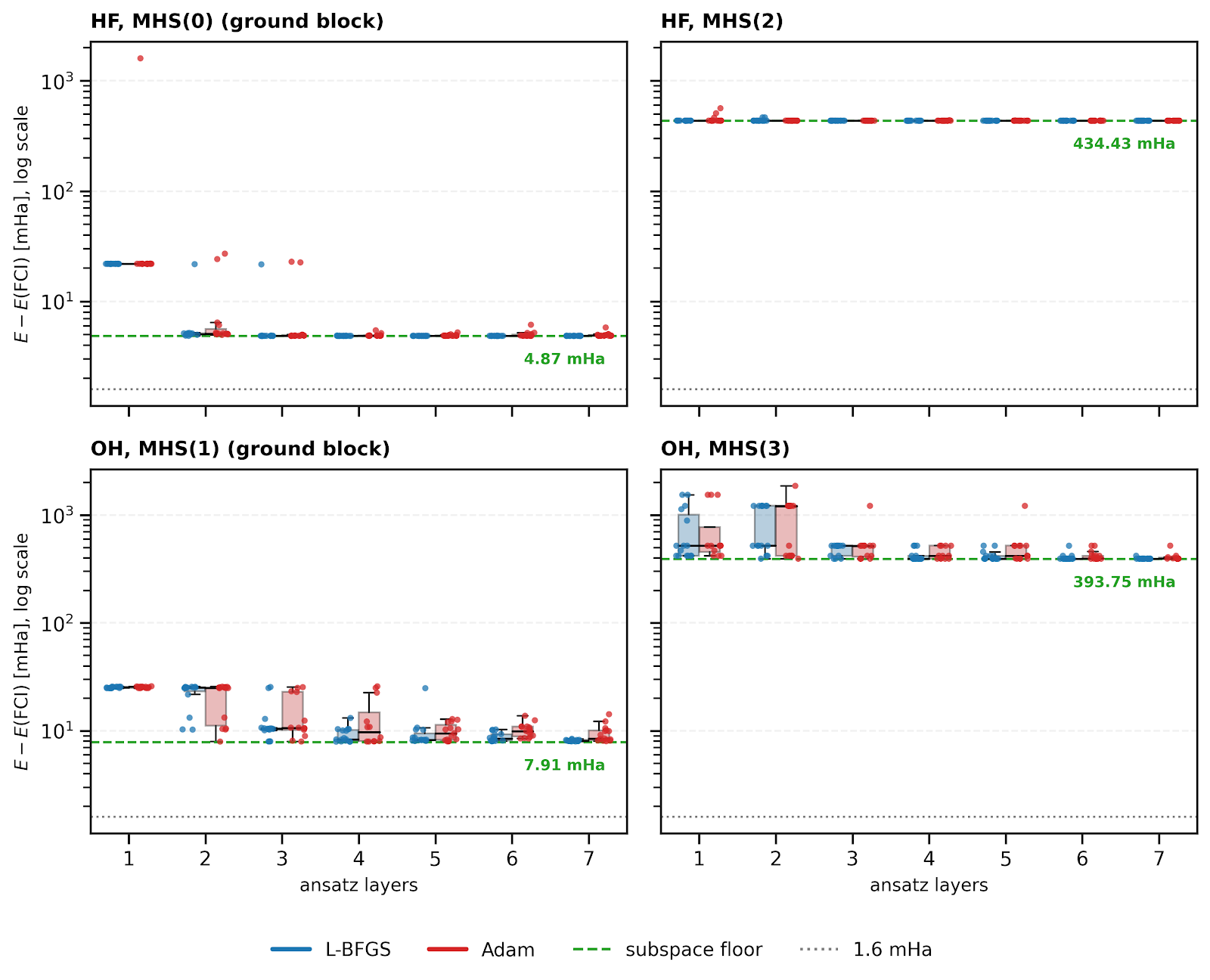}
    \caption{\textcolor{red}{\textit{\small{}VQE convergence for the two-block molecules HF and OH.} \small{}Energy error $E-E_{\mathrm{FCI}}$ (log scale) versus ansatz depth for each seniority block $\mathrm{MHS}(\Omega)$, comparing L-BFGS and Adam over $15$ seeds (boxes: quartiles; whiskers: min--max; points: individual trials). The dashed line marks the exact block ground-state error $E_{\mathrm{MHS}(\Omega)}-E_{\mathrm{FCI}}$; the dotted line marks chemical accuracy ($1.6$~mHa) to FCI.}
    The dashed line marks the exact block ground-state error $E_{\mathrm{MHS}(\Omega)}-E_{\mathrm{FCI}}$; the dotted line marks chemical accuracy relative to FCI ($1.6$~mHa).
    }
    \label{fig:vqe_2block}
\end{figure}

\begin{figure}[h]
    \centering
    \includegraphics[width=\textwidth]{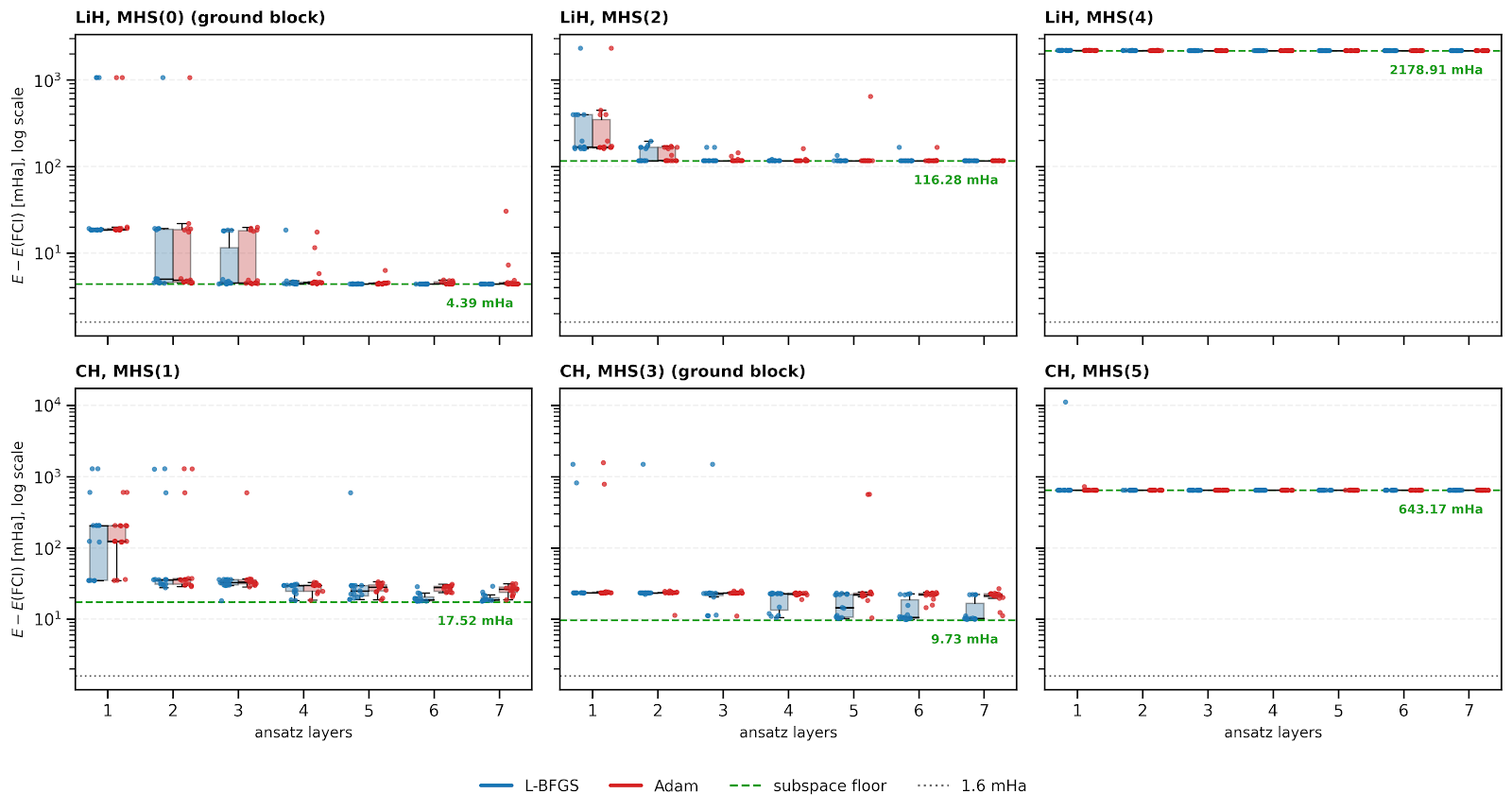}
    \caption{\textcolor{red}{\textit{\small{}VQE convergence for the three-block molecules LiH and CH.} \small{}Energy error $E-E_{\mathrm{FCI}}$ (log scale) versus ansatz depth for each seniority block $\mathrm{MHS}(\Omega)$; conventions as in Figure~\ref{fig:vqe_2block}. For CH, the lowest-energy block is the quartet $\mathrm{MHS}(3)$.}}
    \label{fig:vqe_3block}
\end{figure}

\begin{figure}[h]
    \centering
    \includegraphics[width=0.7\textwidth]{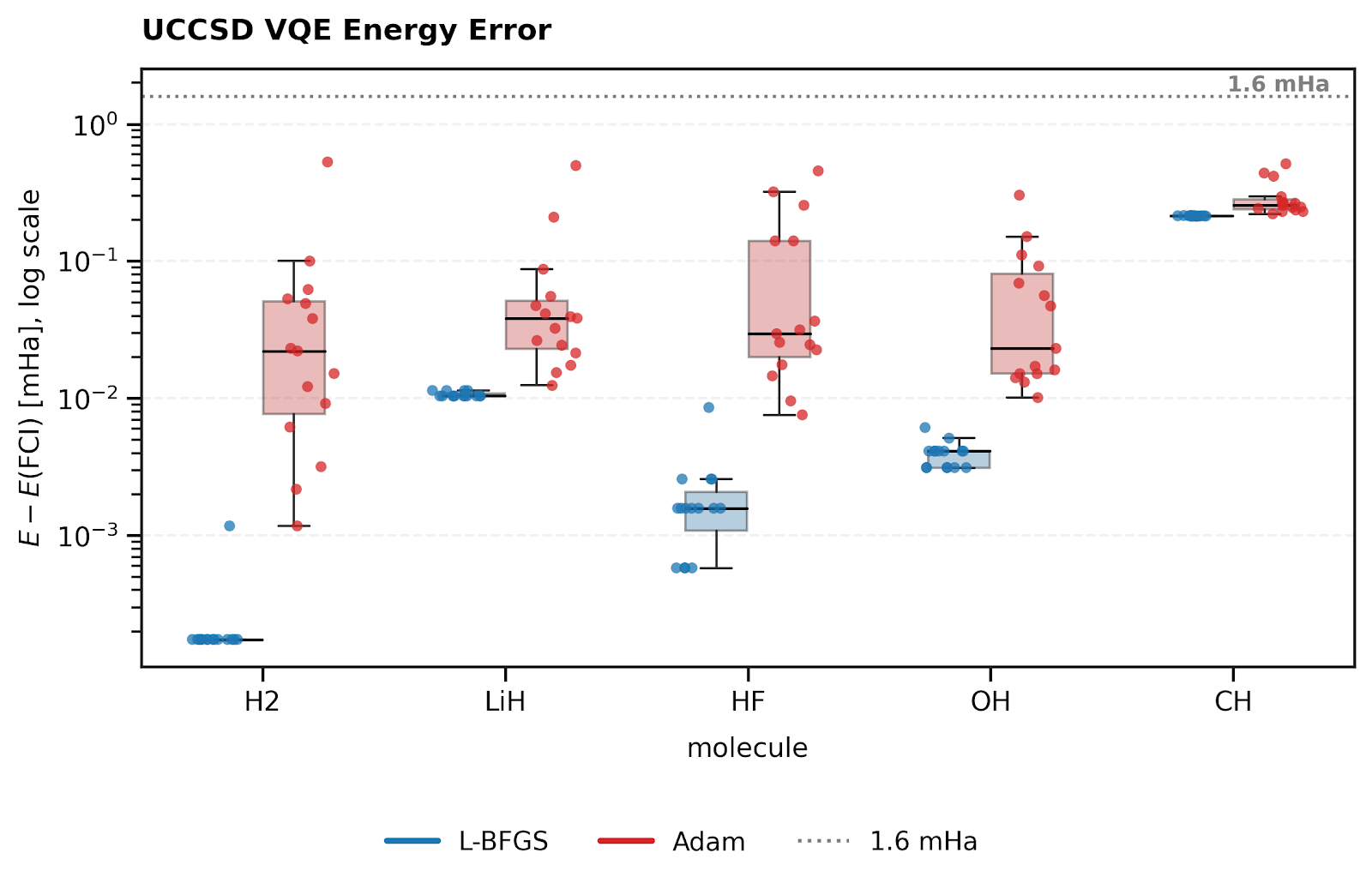}
    \caption{\textcolor{red}{\textit{\small{}UCCSD reference under Jordan--Wigner.} \small{}Energy error $E-E_{\mathrm{FCI}}$ (log scale) per molecule for L-BFGS and Adam over $15$ seeds; all runs reach chemical accuracy ($1.6$~mHa, dotted).}}
    \label{fig:uccsd}
\end{figure}

\textcolor{red}{Within each restricted subspace, the compressed Hamiltonian is generally straightforward to optimize, with circuit depth as the primary control parameter. Across all MHS blocks, the proportion of runs achieving energies within $1.6$~mHa of the exact block ground-state energy $E_{\mathrm{MHS}(\Omega)}$ increases from $27.9\%$ at $L=1$ to $94.5\%$ at $L=7$ for L-BFGS, and from $27.3\%$ to $74.5\%$ for Adam. L-BFGS is therefore the more reliable optimizer, especially for more challenging blocks. Depending on the molecule, the lowest-energy MHS blocks reach their exact block ground-state energies at depths between approximately $L=2$ and $L=6$.}

\textcolor{red}{The primary residual error results from subspace restriction rather than incomplete optimization. For LiH, HF, OH, and CH~\footnote{For CH in the STO-3G basis, FCI yields a high-spin quartet ($S=3/2$, $\langle \widehat{S}^2 \rangle = 3.75$) ground state at $-37.811349\text{ Ha}$. The physical doublet manifold ($^2\Pi$, $S=1/2$, $\langle \widehat{S}^2 \rangle = 0.75$) appears as a two-fold degenerate excited state at $-37.811189\text{ Ha}$, lying merely $0.161\text{ mHa}$ above the quartet minimum due to exchange overstabilization in the STO-3G basis.}, the exact ground-state energies of the lowest-energy MHS blocks are $4.39$, $4.87$, $7.91$, and $9.73$~mHa above the corresponding FCI energies, respectively. Since these errors exceed chemical accuracy, even perfect optimization within the selected blocks cannot achieve chemical accuracy for the true ground-state energies. This limitation aligns with the resource–accuracy trade-off discussed in Sec.~\ref{sec:reference}. The remaining MHS blocks appear at higher energies relative to $E_{\mathrm{FCI}}$, reflecting the energetic separation between different seniority sectors.}

\textcolor{red}{In contrast, the Jordan–Wigner-encoded UCCSD calculations recover the FCI energies to within chemical accuracy for every molecule with both optimizers in Figure~\ref{fig:uccsd}. When applied to the full particle-number-conserving Hilbert space, the UCCSD ansatz accurately approximates the true ground state. However, this accuracy requires substantial circuit resources: the compiled UCCSD circuits use thousands of CNOT gates on twelve qubits in Table~\ref{tab:cnot}. These calculations provide a full-space resource baseline for evaluating the compressed MHS encodings.}

\begin{table}[h]
\scriptsize

\begin{centering}
\begin{tabular}{llrrrr}
\toprule
Molecule & Encoding & $q$ & $T$ & CNOT & CNOT (best) \\
\midrule

\multirow{2}{*}{HF}  & JW     & 12 & 631  & 2400 & 345  \\
                     & MHS(0) & 3  & 36   & 4    & --   \\ \midrule
\multirow{2}{*}{LiH} & JW     & 12 & 631  & 6976 & 1020 \\
                     & MHS(0) & 4  & 134  & 6    & --   \\ \midrule
\multirow{2}{*}{OH}  & JW     & 12 & 631  & 4120 & 624  \\
                     & MHS(1) & 5  & 504  & 16   & --   \\ \midrule
\multirow{2}{*}{CH}  & JW     & 12 & 631  & 7688 & 1165 \\
                     & MHS(3) & 6  & 2072 & 30   & --   \\
\bottomrule
\end{tabular}
\end{centering}
\centering \small
\caption{\textcolor{red}{\emph{\small{}Circuit resources: full-space JW/UCCSD versus the lowest-energy MHS block in STO-3G.} \small{}$q$: qubits; $T$: number of (generalized) Pauli terms (identity included); CNOT: two-qubit gates in the ansatz; ``best'': UCCSD under the best-case qubit-excitation compilation~\cite{Yordanov2020}. MHS rows use the linear-entangler HEA at the median L-BFGS depth reaching $|E_r-E_{\mathrm{MHS}}|\leq1.6$~mHa, i.e.\ $L(q-1)$ CNOTs. Accuracy targets differ: MHS reaches its subspace floor, JW/UCCSD reaches FCI.\label{tab:cnot}}}
\end{table}

\textcolor{red}{Optimization difficulty also varies across seniority sectors. The higher-seniority blocks $\mathrm{MHS}(2)$ for HF, $\mathrm{MHS}(4)$ for LiH, and $\mathrm{MHS}(5)$ for CH reach energies within $1.6$~mHa of their exact block ground-state energies already at $L=1$, whereas the lowest-energy blocks $\mathrm{MHS}(1)$ for OH and $\mathrm{MHS}(3)$ for CH require the maximum circuit depths considered.}

\textcolor{red}{In the present benchmark, the combination of the compact MHS encoding and the shallow HEA requires substantially fewer CNOT gates, the dominant error source on NISQ devices~\cite{Preskill2018, kandala2017hardware}, than the full-space JW/UCCSD calculation. A block of dimension $D$ is encoded using $q=\lceil\log_{2}D\rceil$ qubits. As a result, the lowest-energy blocks require only $3$, $4$, $5$, and $6$ qubits for HF, LiH, OH, and CH, respectively, compared to $12$ qubits for the full active space. At the depths needed to reach the corresponding block ground-state energies, the linear-entangling ansatz uses only $L(q-1)=4$, $6$, $16$, and $30$ CNOT gates, respectively. In contrast, the compiled UCCSD circuits require $2400$, $6976$, $4120$, and $7688$ CNOT gates, or $345$, $1020$, $624$, and $1165$ gates under the best-case qubit-excitation decomposition of Ref.~\cite{Yordanov2020}. The compressed approach reduces the CNOT count by more than two orders of magnitude compared to the compiled circuits, and by more than one order of magnitude even relative to the best-case decomposition. This comparison involves different accuracy targets: MHS-HEA reaches the exact ground-state energy within a selected block, which remains approximately $4$–$10$~mHa above FCI, while UCCSD approximates the FCI ground-state energy itself. The compression therefore trades a few mHa of attainable accuracy, due to subspace restriction, for reduced qubit count and a one- to three-order-of-magnitude reduction in CNOT gates in this VQE benchmark.}

\newpage
\section{Conclusion and Discussion}
\textcolor{red}{
This work develops the Subspace Restriction Scheme (SRS), a framework
for restricting a molecular Hamiltonian to a physically selected
determinant space before encoding it on qubits. We apply SRS to
seniority-resolved sectors selected by a generalized Hund occupation
rule and a fixed spin projection. These sectors are established physical
spaces: $\mathrm{MHS}(0)$ coincides with the seniority-zero DOCI space,
while $\mathrm{MHS}(\Omega)$ describes the corresponding open-shell
seniority sector for $\Omega>0$. The contribution of the present work is
the construction and analysis of compact qubit encodings for these
spaces.}

\textcolor{red}{
The restriction substantially reduces the qubit usage. For
fixed seniority $\Omega$ and filling fraction $c=N/M$, the MHS register
requires $M\mathrm{H}(c/2)+O(\log M)$ qubits, compared with
$2M\mathrm{H}(c/2)+O(\log M)$ for a compact PCS encoding. The $\mathrm{H}_{22}$ calculation illustrates the practical
effect of this reduction: the selected block requires $20$ qubits rather
than the $44$ of the JW encoding, and its sparse representation avoids
constructing a PCS statevector requiring approximately $8.4$~TB in
single precision. This result demonstrates a reduction in qubit
and classical memory, although the reported $\mathrm{H}_{22}$
calculation itself was performed with a classical eigensolver rather
than on quantum hardware.}

\textcolor{red}{
Compact encoding does not reduce every resource simultaneously. The
compressed Hamiltonians can occupy a large fraction of the available
real-symmetric Pauli basis, so matrix sparsity does not imply Pauli
sparsity. At $\Omega=0$, MHS and the hard-core-boson encoding represent
the same DOCI space but exhibit different resource profiles: MHS uses
fewer qubits, whereas HCB generally retains fewer, lower-weight Pauli
terms. Correspondingly, the ratio of the $\Lambda^2/\epsilon^2$
shot-budget bounds increases from $0.40$ for $\mathrm{H}_2$ to $7.91$
for $\mathrm{H}_2\mathrm{O}$. Qubit count, circuit depth, and
measurement cost must therefore be assessed separately when determining
whether compact encoding is advantageous for a particular device.}

\textcolor{red}{
The attainable accuracy is likewise determined by the chosen subspace.
For a variational energy $E_r$ in $\mathrm{MHS}(\Omega)$, the total
error relative to FCI separates into the optimization error
$E_r-E_{\mathrm{MHS}(\Omega)}$ and the irreducible subspace error
$\Delta E(\Omega)=E_{\mathrm{MHS}(\Omega)}-E_{\mathrm{FCI}}$. In the VQE
benchmarks, the compressed blocks are generally straightforward to
optimize, particularly with L-BFGS, and the lowest-energy blocks reach
their respective block ground-state energies using only $4$--$30$ CNOT
gates. Their exact block energies nevertheless remain $4.39$, $4.87$,
$7.91$, and $9.73$~mHa above FCI for LiH, HF, OH, and CH, respectively.
By contrast, the full-space JW/UCCSD calculations recover the FCI
energies to within chemical accuracy for every benchmark molecule, but
require hundreds to thousands of CNOT gates. The two approaches thus
address different accuracy targets: MHS exchanges several mHa of
attainable accuracy for a substantial reduction in qubit requirement and
circuit complexity.}

\textcolor{red}{
The numerical results also identify important physical limitations of
the restriction. Seniority is orbital-dependent, and rotations within
degenerate orbital manifolds can change the MHS energy without changing
the reference SCF energy. When the $\alpha$ and $\beta$ electrons share a common spatial-orbital basis, as in RHF, each $\mathrm{MHS}(\Omega)$ block is spin-pure with $S=\Omega/2$. This property can fail for symmetry-broken UHF orbitals, producing spin contamination. The potential-energy curves of LiH and $\mathrm{N}_2$ further show that the restricted spaces are most reliable near equilibrium and become less accurate in dissociation regimes dominated by static correlation and multireference character.}

\textcolor{red}{
SRS is therefore most useful when qubit requirement is the binding
constraint and the desired state is well represented by a selected
seniority sector. Several extensions may broaden this regime. An
orbital-optimized input, such as one obtained from DOCI- or pCCD-based
orbital optimization, may reduce the orbital dependence and subspace
error~\cite{Bytautas2011,Stein2014}. Meanwhile, a structured labelling of the compact
basis may reduce Pauli density and measurement cost. These directions
preserve the central advantage of SRS---encoding a physically selected
space with the minimum qubit requirement---while addressing the accuracy
and operator-complexity costs exposed by the present benchmarks.}

\section{Data and Software Availability}
The data underlying this study are available in the published article and the Supporting Information in the appendices. The scripts and code used to generate and analyze the results are available on \href{https://github.com/LeoChiang02/subspace_restriction}{GitHub}.

\section{Acknowledgments} The authors thank Ming-Chien Hsu, Kui-Yo Chen, and Yu-Cheng Chen for suggestions, discussions, and technical support. Also, we thank the Quantum Technology Cloud Computing Center, National Cheng Kung University, Taiwan, for supporting our access to  Amazon Web Services AWS. This work is supported by the National Science and Technology Council in Taiwan with grant number MOST 111-2115-M-006 -002 -MY2, NSTC 113-2115-M-006 -015 -MY3, and also internally by the National Cheng Kung University.

\bibliographystyle{ieeetr}
\bibliography{IOP}

\newpage
\appendix
\section{Proofs}\label{appendix:Proofs}

The encoding map $\mathcal{E}$ is defined as a linear transformation that satisfies $\mathcal{E}^{\dagger}\mathcal{E}=I$. Lemma \ref{Lem: Orthonormal Basis}
ensures that generalized Pauli matrices can uniquely decompose any
Hermitian operator.

\begin{lem}
\begin{singlespace}
\label{Lem: Orthonormal Basis} Fix $n\in\mathbb{N}.$ Let $W:=\{w\in M_{2^{n}}(\mathbb{C}):w=w^{\dagger}\}$
be the space of Hermitian matrices. The generalized Pauli matrices, including identity, form an orthonormal basis of $W$ with respect
to the Hilbert-Schmidt inner product $\langle A,B\rangle:=\frac{\mathrm{tr}(A^{\dagger}B)}{2^{n}}$.
\end{singlespace}
\end{lem}

Lemma \ref{Thm: encoding} shows
that the transformation preserves the spectrum and that any Hamiltonian
has a unique representation in the generalized Pauli matrix.

\MyMainThm* \begin{proof}
Clearly $\mathcal{E}^{'}\circ(H|_{\mathcal{F}^{'}}\oplus H^{'})=(H|_{\mathcal{F}^{'}}\oplus H^{'})^{\mathcal{E}^{'}}\circ\mathcal{E}^{'}$ as $(\mathcal{E}^{'})^{\dagger}\mathcal{E}^{'}=I.$ Since $\mathcal{E}^{'}$ is a change of coordinates, $(H|_{\mathcal{F}^{'}}\oplus H^{'})^{\mathcal{E}^{'}}$ is Hermitian.
\end{proof}

\begin{example} Let $\mathcal{F}:=\text{Span}\{|00\rangle_{\text{f}},|01\rangle_{\text{f}},|10\rangle_{\text{f}},|11\rangle_{\text{f}}\}$
and $H:=\sum_{i,j\in\Omega}h(i,j){}_{\text{f}}|i\rangle\langle j|_{\text{f}}$ be a Hamiltonian, where $h(i,j)\in\mathbb{C}$ and $\Omega:=\{00,01,10,11\}.$ Choose $\mathcal{F}^{'}=\text{Span}\{|00\rangle_{\text{f}},|01\rangle_{\text{f}},|10\rangle_{\text{f}}\}$,
then the reduced Hamiltonian $H|_{\mathcal{F}^{'}}=\sum_{i,j\in\Omega'} h(i,j){}_{\text{f}}|i\rangle\langle j|_{\text{f}}\ {\rm with}\ \Omega^{'}:=\{00,01,10\}.$ Here, the restriction operator is $P_{\mathcal{F}'}=\sum_{ i\in\Omega'}{}_{\text{f}}|i\rangle\langle i|_{\text{f}}$. A minimal extension of $H|_{\mathcal{F}^{'}}$ can be constructed as  $\sum _{i,j\in\Omega'}h(i,j){}_{\text{f}}|i\rangle\langle j|_{\text{f}}\oplus0_{V}$ with $0_{V}$ the zero map on $V:=\text{Span}\{v\}$.

Assume further that $H|_{\mathcal{F}^{'}}:={}_{\text{f}}|01\rangle\langle10|_{\text{f}}+{}_{\text{f}}|10\rangle\langle01|_{\text{f}}$
and let $\mathcal{Q}:=\text{Span}\{|00\rangle_{\text{q}},|01\rangle_{\text{q}},|10\rangle_{\text{q}},|11\rangle_{\text{q}}\}.$
For an encoding map $\mathcal{E}^{'}:\mathcal{F}'\oplus V\rightarrow Q$ defined by
\[ \mathcal{E}'(x)=
\begin{cases}
|00\rangle_{\text{q}}, & x=|00\rangle_{\text{f}};\\
|01\rangle_{\text{q}}, & x=|01\rangle_{\text{f}};\\
|10\rangle_{\text{q}}, & x=|10\rangle_{\text{f}};\\
|11\rangle_{\text{q}}, & x=v,
\end{cases}
\]
the encoded reduced Hamiltonian has a representation in generalized Pauli's as  $I_{0}X_{1}+Z_{0}X_{1}$.
\end{example}

\begin{example} When we choose the target subspace $\mathcal{F}'$ as the entire Fock space and  $\mathcal{E}'(|i\rangle_{\rm f})=|i\rangle_{\rm q}$ as the identity map. This method is equivalent to the JW transformation.
\end{example}

This appendix provides the proofs of the encoding properties, the
seniority-block decomposition, the spin-purity result, the dimension
estimates, and the Pauli-term ceiling used in the main text.

\SeniorityBlock*
\begin{proof}[Proof of Corollary~\ref{cor:block}]
Every determinant in $\mathrm{MHS}(\Omega)$ contains $\Omega$ singly
occupied $\alpha$ orbitals and therefore satisfies
\[\widehat S_z|\Phi\rangle=\frac{\Omega}{2}|\Phi\rangle.\]
Let
$|\Phi\rangle\in\mathrm{MHS}(\Omega)$ and
$|\Psi\rangle\in\mathrm{MHS}(\Omega')$. Since $\widehat H$ commutes with $\widehat S_z$. Consequently,
\[0=\langle\Phi|[\widehat S_z,\widehat H]|\Psi\rangle=\frac{\Omega-\Omega'}{2}\langle\Phi|\widehat H|\Psi\rangle.\]
If $\Omega\neq\Omega'$, this implies
\[\langle\Phi|\widehat H|\Psi\rangle=0.\]
Thus, the projected Hund-space Hamiltonian is block diagonal:
\[P_{\mathrm{HS}}\widehat H P_{\mathrm{HS}}=\bigoplus_{\Omega}P_{\mathrm{MHS}(\Omega)}\widehat HP_{\mathrm{MHS}(\Omega)}.
\]
The smallest eigenvalue of a direct sum is the minimum of the smallest
eigenvalues of its blocks, so
\[E_{\mathrm{HS}}=\min_{\Omega}E_{\mathrm{MHS}(\Omega)}.\]
\end{proof}

\SizeEstimation* \begin{proof}
Let \textcolor{red}{$k\in\left\{ 0,\ldots,\left\lfloor \frac{N}{2}\right\rfloor \right\}$}. Denote $N-k$ and $k$ as the number of spin-up
and spin-down electrons, respectively. Since we place a spin-down
electron in an orbital, only when there exists a spin-up electron,
the combination in the $M$ orbitals system is $\binom{M}{N-k}\tbinom{N-k}{k}$.
Furthermore, the generalized Hund's rule implies that \textcolor{red}{$N\ge2k$}. Therefore,
the number of fermionic bases satisfying the generalized Hund's rule
is $\sum\limits_{k=0}^{\lfloor\frac{N}{2}\rfloor}\binom{M}{N-k}\tbinom{N-k}{k}$.
\end{proof}

\SpinPurity*
\begin{proof}
Let
\[
\Sigma_{pq}:=\langle\phi_p^\alpha|\phi_q^\beta\rangle,
\qquad
n_{p\sigma}:=a_{p\sigma}^{\dagger}a_{p\sigma}
\quad (\sigma=\alpha,\beta).
\]
The spin operators are~\cite{Helgaker2000}
\[
\widehat S_+
=
\sum_{p,q}\Sigma_{pq}a_{p\alpha}^{\dagger}a_{q\beta},
\qquad
\widehat S_-=\widehat S_+^\dagger,
\]
\[
\widehat S_z
=
\frac12\sum_p\left(n_{p\alpha}-n_{p\beta}\right),
\qquad
\widehat S^2
=
\widehat S_-\widehat S_+
+
\widehat S_z(\widehat S_z+1).
\]

For a Hund-compatible determinant $|\Phi\rangle$ of seniority $\Omega$,
\[
\widehat S_z|\Phi\rangle
=
\frac{\Omega}{2}|\Phi\rangle.
\]
Moreover, the fermionic anticommutation relations give
\[
\begin{aligned}
\langle\Phi|\widehat S_-\widehat S_+|\Phi\rangle
&=
\sum_q
\langle\Phi|n_{q\beta}|\Phi\rangle \\
&\quad -
\sum_{p,q}|\Sigma_{pq}|^2
\langle\Phi|n_{p\alpha}n_{q\beta}|\Phi\rangle \\
&=
N_\beta
-
\sum_{p\in\mathcal P\cup\mathcal U}
\sum_{q\in\mathcal P}
|\Sigma_{pq}|^2.
\end{aligned}
\]
Therefore,
\[
\langle\Phi|\widehat S^2|\Phi\rangle
=
\frac{\Omega}{2}
\left(\frac{\Omega}{2}+1\right)
+
N_\beta
-
\sum_{p\in\mathcal P\cup\mathcal U}
\sum_{q\in\mathcal P}
|\Sigma_{pq}|^2.
\]
Under hypothesis~\textup{(H1)}, $\Sigma_{pq}=\delta_{pq}$, and hence
\[
\sum_{p\in\mathcal P\cup\mathcal U}
\sum_{q\in\mathcal P}
|\Sigma_{pq}|^2
=
\sum_{q\in\mathcal P}1
=
|\mathcal P|
=
N_\beta.
\]
Thus,
\[
\widehat S^2|\Phi\rangle
=
\frac{\Omega}{2}
\left(\frac{\Omega}{2}+1\right)
|\Phi\rangle.
\]
Every state in $\mathrm{MHS}(\Omega)$ is a linear combination of such
determinants, so the entire block has total spin $S=\Omega/2$.
\end{proof}

\MHS*

\begin{proof}[Proof of Corollary~\ref{Multi+hnud}]
For a target block with $\Omega=2S$, the electron numbers are
\[
N_\alpha=\frac{N+2S}{2},
\qquad
N_\beta=\frac{N-2S}{2}.
\]
We first choose the $N_\alpha$ orbitals occupied by $\alpha$ electrons
and then choose the $N_\beta$ paired orbitals from this set. Hence
\begin{align*}
\dim\mathrm{MHS}(\Omega) &= \binom{M}{\frac{N+2S}{2}} \binom{\frac{N+2S}{2}}{\frac{N-2S}{2}} \\
&= \frac{M!}{(2S)!\bigl(\tfrac{N}{2}-S\bigr)!\bigl(M-\tfrac{N}{2}-S\bigr)!} =: \binom{M}{2S, \tfrac{N}{2}-S, M-\tfrac{N}{2}-S}.
\end{align*}
Under hypothesis~\textup{(H1)}, Corollary~\ref{cor:spin} shows that this block has total spin \(S=\Omega/2\) and multiplicity 2S+1
\end{proof}

\QubitSaving*

\begin{proof} Let $\mathrm{H}(x)=-x\log_2x-(1-x)\log_2(1-x)$ and $\mathrm{H}_3(x_1,x_2,x_3)=-\sum x_i\log_2x_i$. Put $c=N/M$, $\sigma=S/M$, and

\[
x_U=2\sigma,\qquad x_D=\tfrac c2-\sigma,\qquad x_E=1-\tfrac c2-\sigma \qquad (x_U+x_D+x_E=1).\]
We claim:
\[
\frac{2^{M\mathrm{H}_3(x_U,x_D,x_E)}}{(M+1)^3}\;\le\;\mathrm{MHS}(\Omega)\;\le\;2^{M\mathrm{H}_3(x_U,x_D,x_E)};
\]
\[\frac{2^{2M\mathrm{H}(c/2)}}{2M+1}\;\le\;\binom{2M}{cM} = \dim\mathrm{PCS}\;\le\;2^{2M\mathrm{H}(c/2)} .\]
Upper bound of MHS($\Omega$):
\begin{align*}
1 = (x_{U}+x_{D}+x_{E})^{M}&=\sum_{\substack{j_{1},j_{2},j_{3}\ge0\\
j_{1}+j_{2}+j_{3}=M
}
}\binom{M}{j_{1},j_{2},j_{3}}x_{U}^{j_{1}}x_{D}^{j_{2}}x_{E}^{j_{3}}\\&\ge\binom{M}{Mx_{U},Mx_{D},Mx_{E}}x_{U}^{Mx_{U}}x_{D}^{Mx_{D}}x_{E}^{Mx_{E}}.\\&\ge \dim {\text{MHS}(\Omega)}\cdot2^{\log_{2}\left(x_{U}^{Mx_{U}}x_{D}^{Mx_{D}}x_{E}^{Mx_{E}}\right)}=\dim {\text{MHS}(\Omega)}\cdot2^{-M\mathrm{H}_{3}(x_{U},x_{D},x_{E})}
\end{align*}Lower bound of MHS:
For $\mathrm{MHS}(\Omega)$, the trinomial coefficient is the size of a type class over an alphabet of size $\vert{}\mathcal{X}\vert{} = 3$ with block length $M$, so applying Theorem 11.1.3 in \cite{InformationTheory} yields the lower bound with denominator $(M+1)^3$.

For $\dim \mathrm{PCS} = \binom{2M}{cM}$, the binary alphabet case with $n = 2M$ and $k = cM$ gives the both upper and lower bound by Example 11.1.3, Eq. (11.40) in ~\cite{InformationTheory}. At fixed $S$ (so $\sigma\to0$), we have $\mathrm{H}_3(x_U,x_D,x_E)\to \mathrm{H}(c/2)$. Therefore,
\[\bigl\lceil\log_2 \dim\mathrm{MHS}(\Omega)\bigr\rceil \;=\; M\cdot\mathrm{H}(c/2)+O(\log M) \approx \frac{1} {2}\cdot \bigl\lceil\log_2 \dim \mathrm{PCS}\bigr\rceil.\]

\end{proof}

\begin{prop}
\label{prop:pauli-ceiling}
Let $A$ be a real symmetric operator acting on $q$ qubits. Then its
Pauli decomposition contains at most
\[\frac{4^q+2^q}{2}\]
nonzero terms.
\end{prop}

\begin{proof}
The coefficient of a Pauli string $P$ is $c_P=2^{-q}\operatorname{tr}(PA)$.
If $P$ contains an odd number of $Y$ operators, then $P^T=-P$.
Since $A^T=A$,
$
\operatorname{tr}(PA)
=
\operatorname{tr}\bigl((PA)^T\bigr)
=
\operatorname{tr}(AP^T)
=
-\operatorname{tr}(AP)
=
-\operatorname{tr}(PA),
$
so $c_P=0$. Therefore, only Pauli strings containing an even number of
$Y$ operators can have nonzero coefficients. The number of such strings is
\[
\sum_{\substack{k=0\\k\ \mathrm{even}}}^{q}
\binom{q}{k}3^{q-k}
=
\frac{(3+1)^q+(3-1)^q}{2}
=
\frac{4^q+2^q}{2}.
\]
This gives the stated ceiling.
\end{proof}
Finally, because
$
\mathrm{MHS}(\Omega)\subseteq\mathrm{PCS},
$
the variational principle gives
$
E_{\mathrm{MHS}(\Omega)}
\geq
E_{\mathrm{FCI}}.
$
Consequently, the irreducible subspace error
$
\Delta E(\Omega)
:=
E_{\mathrm{MHS}(\Omega)}-E_{\mathrm{FCI}}
$
is nonnegative. If a normalized VQE state within
$\mathrm{MHS}(\Omega)$ returns the energy $E_r$, then
\[
E_r-E_{\mathrm{FCI}}
=
\underbrace{
E_r-E_{\mathrm{MHS}(\Omega)}
}_{\epsilon_{\mathrm{opt}}}
+
\underbrace{
E_{\mathrm{MHS}(\Omega)}-E_{\mathrm{FCI}}
}_{\Delta E(\Omega)}.
\]
Both terms are nonnegative: $\epsilon_{\mathrm{opt}}$ is the
optimization error, while $\Delta E(\Omega)$ is the irreducible error
introduced by restricting the accessible subspace.

\newpage
\section{Algorithm}\label{appendix:Algorithm}

In algorithm 1, $j$, $n$, and $\bigtriangleup$ represent acting
orbital, Fock space dimensions, and the operator type, respectively.
The first output is the phase, and the second is the corresponding
vector after operating.

\begin{algorithm}
\caption{}

\begin{algorithmic}[1]
\State input: $j,n,\bigtriangleup$
\\
\State $B \gets $ ARRAY[1,$2^n$][1,n]
\State $v \gets $ ARRAY[1,n] of 0
\\
\While{$i<j$}
	\State $v[i] \leftarrow 1$
\EndWhile
\\
\For{$i \gets 1$ \textbf{to} $2^n$}
	\For{$j \gets 1$ \textbf{to} $n$}
		\State $ B[i][j] \leftarrow (i>>j)\& 1 $
	\EndFor
\EndFor
\\
\State $P \gets Bv$
\\
\If {$\bigtriangleup$ is creation}
	\State $L \gets (B[:][j] == 0)$
	\State $B[:][j] \gets 1$
\ElsIf {$\bigtriangleup$ is annhilation}
	\State $L \gets (B[:][j] == 1)$
	\State $B[:][j] \gets 0$
\EndIf
\\
\State output: $(-1)^{P[L]}, B[L]$
\end{algorithmic}
\end{algorithm}

\section{Device Details}

\begin{table}[H]
\begin{tabular}{cccc}
\hline
Task            & Hardware                                                      & Backend                   & Key Parameters \\ \hline
Non-VQE part    & \begin{tabular}[c]{@{}c@{}}\textcolor{red}{m7i.8xlarge}\cite{braket}\end{tabular} & Classical Eigensolver\cite{cupy_learningsys2017}    & float64              \\
\textcolor{red}{VQE (MHS)} & Intel 13900H                                             & Pennylane lightning-qubit\cite{PennyLane}
& 1000 shots, \textcolor{red}{noiseless}     \\
\textcolor{red}{VQE (JW,UCCSD)}       & Intel E5-2698 v4                                                          & Pennylane lightning-qubit\cite{PennyLane}
& 1000 shots, \textcolor{red}{noiseless}       \\ \hline
\end{tabular}
\end{table}

\newpage
\section{Simulation Data in the STO-3G basis set}\label{appendix:Data}

\begin{table}[!ht]
\centering
\footnotesize 
\begin{tabular}{l cc cc cc cc cc}
\toprule
\multirow{2}{*}{Molecule} & \multirow{2}{*}{$N_{so}$} & \multirow{2}{*}{$N_{e}$} & \multicolumn{2}{c}{PCS} & \multicolumn{2}{c}{MS} & \multicolumn{2}{c}{HS} & \multicolumn{2}{c}{MHS} \\
\cmidrule(lr){4-5} \cmidrule(lr){6-7} \cmidrule(lr){8-9} \cmidrule(lr){10-11}
 &  &  & Size & Qubits & Size & Qubits & Size & Qubits & Size & Qubits \\
\midrule
$\mathrm{H_{2}}$   & 4  & \textbf{2}  & 6      & 3  & 4      & 2  & 3     & 2  & 2   & 1 \\
$\mathrm{LiH}$     & 12 & \textbf{4}  & 495    & 9  & 225    & 8  & 90    & 7  & 15  & 4 \\
$\mathrm{BeH_{2}}$ & 14 & \textbf{6}  & 3003   & 12 & 1225   & 11 & 357   & 9  & 35  & 6 \\
$\mathrm{CH}$      & 12 & \textbf{7}  & 792    & 10 & 300    & 9  & 126   & 7  & 60  & 6 \\
$\mathrm{NH}$      & 12 & \textbf{8}  & 495    & 9  & 120    & 7  & 90    & 7  & 60  & 6 \\
$\mathrm{BH_{3}}$  & 16 & \textbf{8}  & 12870  & 14 & 4900   & 13 & 1107  & 11 & 70  & 7 \\
$\mathrm{OH}$      & 12 & \textbf{9}  & 220    & 8  & 90     & 7  & 50    & 6  & 30  & 5 \\
$\mathrm{HF}$      & 12 & \textbf{10} & 66     & 7  & 36     & 6  & 21    & 5  & 6   & 3 \\
$\mathrm{H_{2}O}$  & 14 & \textbf{10} & 1001   & 10 & 441    & 9  & 161   & 8  & 21  & 5 \\
$\mathrm{NH_{3}}$  & 16 & \textbf{10} & 8008   & 13 & 3136   & 12 & 784   & 10 & 56  & 6 \\
$\mathrm{CH_{4}}$  & 18 & \textbf{10} & 43758  & 16 & 15876  & 14 & 2907  & 12 & 126 & 7 \\
$\mathrm{CN}$      & 20 & \textbf{13} & 77520  & 17 & 25200  & 15 & 4740  & 13 & 840 & 10 \\
$\mathrm{N_{2}}$   & 20 & \textbf{14} & 38760  & 16 & 14400  & 14 & 2850  & 12 & 120 & 7 \\
$\mathrm{CO}$      & 20 & \textbf{14} & 38760  & 16 & 14400  & 14 & 2850  & 12 & 120 & 7 \\
$\mathrm{CNH}$     & 22 & \textbf{14} & 319770 & 19 & 108900 & 17 & 14355 & 14 & 330 & 9 \\
$\mathrm{NO}$      & 20 & \textbf{15} & 15504  & 14 & 5400   & 13 & 1452  & 11 & 360 & 9 \\
$\mathrm{O_{2}}$   & 20 & \textbf{16} & 4845   & 13 & 1200   & 11 & 615   & 10 & 360 & 9 \\
$\mathrm{NF}$      & 20 & \textbf{16} & 4845   & 13 & 1200   & 11 & 615   & 10 & 360 & 9 \\
$\mathrm{H_{2}CO}$ & 24 & \textbf{16} & 735471 & 20 & 245025 & 18 & 28314 & 15 & 495 & 9 \\
$\mathrm{H_{2}O_{2}}$ & 24 & \textbf{18} & 134596 & 18 & 48400  & 16 & 8074  & 13 & 220 & 8 \\
$\mathrm{H_{22}}$ & 44 & \textbf{22} & $\sim10^{12}$*  & 41* & $\sim10^{11}$* & 39* & $\sim10^{9}$* & 32* & 705432 & 20 \\
\bottomrule
\end{tabular}
\caption{Tested molecular system size and qubit requirements across different reduction methods in the STO-3G basis set. The asterisk indicates that the molecule-subspace combination has not been tested. $N_{so}$ and $N_{e}$ refer to the number of spin orbitals and electrons, respectively.}\label{tab:systemsize}
\end{table}


\begin{table}[tbp]
\caption{STO-3G Ground State Results for the $\text{H}_{22}$ with \textbf{MHS}.}
\label{tab:H22}
The $H-H$ bond length in $\text{H}_{22}$ is selected as 1 Angstrom.
\begin{tabular}{ccccc}
\hline
$D$    & $E_{\text{MHS}}$ & $E_{\text{RHF}}$ & Eigensolver & Time Consumption \\ \hline
705432 & -11.48939 Ha & -11.45729 Ha                      & CuPy eigsh solver \cite{cupy_learningsys2017} & $<$30 mins    \\ \hline
\end{tabular}

\end{table}

\begin{table}[p]
\centering
\scriptsize
\setlength{\tabcolsep}{3pt}
\renewcommand{\arraystretch}{0.92}

\caption{\emph{Spin purity of the $\mathrm{MHS}(\Omega)$ blocks
(STO-3G).} For each block, $S=\Omega/2$ and the target eigenvalue is
$S(S+1)$. The final two columns give
$\langle\widehat S^2\rangle$ for the block ground state under the RHF
and UHF references, respectively. Under RHF, every block is spin-pure,
whereas UHF-based blocks may exhibit spin contamination.}
\label{tab:spin_purity}

\begin{minipage}[t]{0.49\textwidth}
\vspace{0pt}
\centering
\begin{tabular}{lccccc}
\toprule
Molecule & $\Omega$ & $S$ & $S(S+1)$ & RHF & UHF \\
\midrule

\multirow{2}{*}{$\mathrm{H_2}$}
 & 0 & $0$ & 0.00 & 0.0000 & 0.0000 \\
 & 2 & $1$ & 2.00 & 2.0000 & 2.0000 \\
\addlinespace

\multirow{3}{*}{$\mathrm{LiH}$}
 & 0 & $0$ & 0.00 & 0.0000 & 0.0000 \\
 & 2 & $1$ & 2.00 & 2.0000 & 2.0000 \\
 & 4 & $2$ & 6.00 & 6.0000 & 6.0000 \\
\addlinespace

\multirow{4}{*}{$\mathrm{BeH_2}$}
 & 0 & $0$ & 0.00 & 0.0000 & 0.0032 \\
 & 2 & $1$ & 2.00 & 2.0000 & 2.0000 \\
 & 4 & $2$ & 6.00 & 6.0000 & 6.0000 \\
 & 6 & $3$ & 12.00 & 12.0000 & 12.0000 \\
\addlinespace

\multirow{3}{*}{$\mathrm{CH}$}
 & 1 & $1/2$ & 0.75 & 0.7500 & 1.3548 \\
 & 3 & $3/2$ & 3.75 & 3.7500 & 3.7514 \\
 & 5 & $5/2$ & 8.75 & 8.7500 & 8.7500 \\
\addlinespace

\multirow{3}{*}{$\mathrm{NH}$}
 & 0 & $0$ & 0.00 & 0.0000 & 1.1518 \\
 & 2 & $1$ & 2.00 & 2.0000 & 2.0043 \\
 & 4 & $2$ & 6.00 & 6.0000 & 6.0000 \\
\addlinespace

\multirow{5}{*}{$\mathrm{BH_3}$}
 & 0 & $0$ & 0.00 & 0.0000 & 0.0000 \\
 & 2 & $1$ & 2.00 & 2.0000 & 2.0000 \\
 & 4 & $2$ & 6.00 & 6.0000 & 6.0000 \\
 & 6 & $3$ & 12.00 & 12.0000 & 12.0000 \\
 & 8 & $4$ & 20.00 & 20.0000 & 20.0000 \\
\addlinespace

\multirow{2}{*}{$\mathrm{OH}$}
 & 1 & $1/2$ & 0.75 & 0.7500 & 0.7529 \\
 & 3 & $3/2$ & 3.75 & 3.7500 & 3.7500 \\
\addlinespace

\multirow{2}{*}{$\mathrm{HF}$}
 & 0 & $0$ & 0.00 & 0.0000 & 0.0001 \\
 & 2 & $1$ & 2.00 & 2.0000 & 2.0000 \\
\addlinespace

\multirow{3}{*}{$\mathrm{H_2O}$}
 & 0 & $0$ & 0.00 & 0.0000 & 0.0000 \\
 & 2 & $1$ & 2.00 & 2.0000 & 2.0000 \\
 & 4 & $2$ & 6.00 & 6.0000 & 6.0000 \\
\addlinespace

\multirow{4}{*}{$\mathrm{NH_3}$}
 & 0 & $0$ & 0.00 & 0.0000 & 0.0000 \\
 & 2 & $1$ & 2.00 & 2.0000 & 2.0000 \\
 & 4 & $2$ & 6.00 & 6.0000 & 6.0000 \\
 & 6 & $3$ & 12.00 & 12.0000 & 12.0000 \\
\addlinespace

\multirow{5}{*}{$\mathrm{CH_4}$}
 & 0 & $0$ & 0.00 & 0.0000 & 0.0037 \\
 & 2 & $1$ & 2.00 & 2.0000 & 2.0035 \\
 & 4 & $2$ & 6.00 & 6.0000 & 6.0016 \\
 & 6 & $3$ & 12.00 & 12.0000 & 12.0000 \\
 & 8 & $4$ & 20.00 & 20.0000 & 20.0000 \\

\bottomrule
\end{tabular}
\end{minipage}
\hfill
\begin{minipage}[t]{0.49\textwidth}
\vspace{0pt}
\centering
\begin{tabular}{lccccc}
\toprule
Molecule & $\Omega$ & $S$ & $S(S+1)$ & RHF & UHF \\
\midrule

\multirow{4}{*}{$\mathrm{CN}$}
 & 1 & $1/2$ & 0.75 & 0.7500 & 1.2651 \\
 & 3 & $3/2$ & 3.75 & 3.7500 & 3.9992 \\
 & 5 & $5/2$ & 8.75 & 8.7500 & 8.7581 \\
 & 7 & $7/2$ & 15.75 & 15.7500 & 15.7500 \\
\addlinespace

\multirow{4}{*}{$\mathrm{N_2}$}
 & 0 & $0$ & 0.00 & 0.0000 & 0.0250 \\
 & 2 & $1$ & 2.00 & 2.0000 & 2.0002 \\
 & 4 & $2$ & 6.00 & 6.0000 & 6.0001 \\
 & 6 & $3$ & 12.00 & 12.0000 & 12.0000 \\
\addlinespace

\multirow{4}{*}{$\mathrm{CO}$}
 & 0 & $0$ & 0.00 & 0.0000 & 0.0100 \\
 & 2 & $1$ & 2.00 & 2.0000 & 2.0311 \\
 & 4 & $2$ & 6.00 & 6.0000 & 6.0019 \\
 & 6 & $3$ & 12.00 & 12.0000 & 12.0000 \\
\addlinespace

\multirow{5}{*}{$\mathrm{CNH}$}
 & 0 & $0$ & 0.00 & 0.0000 & 0.0269 \\
 & 2 & $1$ & 2.00 & 2.0000 & 2.0168 \\
 & 4 & $2$ & 6.00 & 6.0000 & 6.0004 \\
 & 6 & $3$ & 12.00 & 12.0000 & 12.0001 \\
 & 8 & $4$ & 20.00 & 20.0000 & 20.0000 \\
\addlinespace

\multirow{3}{*}{$\mathrm{NO}$}
 & 1 & $1/2$ & 0.75 & 0.7500 & 0.7986 \\
 & 3 & $3/2$ & 3.75 & 3.7500 & 3.7523 \\
 & 5 & $5/2$ & 8.75 & 8.7500 & 8.7500 \\
\addlinespace

\multirow{3}{*}{$\mathrm{O_2}$}
 & 0 & $0$ & 0.00 & 0.0000 & 1.9331 \\
 & 2 & $1$ & 2.00 & 2.0000 & 2.0096 \\
 & 4 & $2$ & 6.00 & 6.0000 & 6.0000 \\
\addlinespace

\multirow{3}{*}{$\mathrm{NF}$}
 & 0 & $0$ & 0.00 & 0.0000 & 0.1732 \\
 & 2 & $1$ & 2.00 & 2.0000 & 2.0109 \\
 & 4 & $2$ & 6.00 & 6.0000 & 6.0000 \\
\addlinespace

\multirow{5}{*}{$\mathrm{H_2CO}$}
 & 0 & $0$ & 0.00 & 0.0000 & 0.0000 \\
 & 2 & $1$ & 2.00 & 2.0000 & 2.0000 \\
 & 4 & $2$ & 6.00 & 6.0000 & 6.0000 \\
 & 6 & $3$ & 12.00 & 12.0000 & 12.0000 \\
 & 8 & $4$ & 20.00 & 20.0000 & 20.0000 \\
\addlinespace

\multirow{4}{*}{$\mathrm{H_2O_2}$}
 & 0 & $0$ & 0.00 & 0.0000 & 0.0000 \\
 & 2 & $1$ & 2.00 & 2.0000 & 2.0000 \\
 & 4 & $2$ & 6.00 & 6.0000 & 6.0000 \\
 & 6 & $3$ & 12.00 & 12.0000 & 12.0000 \\

\bottomrule
\end{tabular}
\end{minipage}

\end{table}

\clearpage

\end{document}